\documentclass{llncs}
\usepackage{xcolor} 
\usepackage{amsmath}
\usepackage{amssymb}
\usepackage{proof}
\usepackage{hyperref}
\usepackage{caption}
\usepackage{comment}
\usepackage{subcaption}
\usepackage{fancybox}
\usepackage{stmaryrd}
\definecolor{light-gray}{gray}{0.86}

\newcommand{\interp}[1]{\llbracket #1 \rrbracket}

\usepackage{todonotes}
\newcommand{\knote}[1]{\todo[inline, color=blue!20]{#1}}

\begin{document}
\title{A Type-Theoretic Approach to Structural Resolution}

\author{Peng Fu, Ekaterina Komendantskaya}
\institute{School of Computing, University of Dundee}
%\date{\today}

\maketitle

\begin{abstract}
 %  Structural resolution (or S-Resolution) is a variant of standard SLD-resolution, it proposes to distinguish \textit{productive} logic programs from the unproductive ones. Productive programs allow finite
 %  observation of partial answers. 
 %   In this paper, we propose a type system for the analysis of S-resolution and SLD-resolution. 
 %  %The three main contributions of this paper are
 %  We formulate S-resolution and SLD-resolution as reduction systems, and show their soundness relative to the type system.
 %  One of the central method of this paper is \textit{realizability transformation}, which makes logic programs productive and non-overlapping.
 %  %, which adds an extra argument to each atomic formula to
 % % reflect its proof content.
 %  We show that S-resolution and SLD-resolution are only equivalent for programs with these two properties. %Moreover, the transformation gives 
Structural resolution (or S-resolution) is a newly proposed alternative to SLD-resolution that 
allows a systematic separation of derivations into term-matching and unification steps. 
Productive logic programs are those for which term-matching reduction on any query must terminate. For productive programs with coinductive meaning, finite term-rewriting reductions can be seen as measures of observation in an infinite derivation. Ability of handling corecursion in a productive way is an attractive computational feature of S-resolution.

 In this paper, we make first steps towards a better conceptual understanding of operational properties of S-resolution as compared to SLD-resolution. To this aim, we propose a type system for the analysis of both SLD-resolution and S-resolution.  
  We formulate S-resolution and SLD-resolution as reduction systems, and show their soundness relative to the type system.
  One of the central methods of this paper is \textit{realizability transformation}, which makes logic programs productive and non-overlapping.
  %, which adds an extra argument to each atomic formula to
 % reflect its proof content.
  We show that S-resolution and SLD-resolution are only equivalent for programs with these two properties. %Moreover, the transformation gives 


\end{abstract}
\section{Introduction}
\label{introduction}

Logic Programming (LP) is a programming
paradigm based on first-order Horn formulas. Informally, given a
logic program $\Phi$ and a query $A$, LP
provides a mechanism for automatically inferring whether or not $\Phi
\vdash A$ holds, i.e., whether or not
$\Phi$ logically entails $A$. The mechanism for the logical inference is based on SLD-resolution algorithm,
which uses the resolution rule together with first-order unification.

\begin{example}\label{ex:conn}
Consider the following logic program  $\Phi$, consisting of Horn formulas labelled by $\kappa_1$, $\kappa_2$, $\kappa_3$, defining connectivity for a graph with three nodes:
 {\footnotesize
  \begin{center}
  $  \kappa_1 : \forall x. \forall y. \forall z. \mathrm{Connect}(x, y), \mathrm{Connect}(y, z) \Rightarrow \mathrm{Connect}(x, z)$

  $\kappa_2 : \ \Rightarrow \mathrm{Connect}(\mathrm{node_1}, \mathrm{node_2})$
  
  $\kappa_3 : \ \Rightarrow \mathrm{Connect}(\mathrm{node_2}, \mathrm{node_3})$
    \end{center}
}  
  
  \noindent  %We annotate each program by $\kappa_1,\kappa_2, \kappa_3$; we use $\Phi$ to denote the logic program above;
  In the above program, $\mathrm{Connect}$ is a predicate, and  $\mathrm{node_1}$ --  $\mathrm{node_3}$ are constants.
 SLD-derivation for  the query $\mathrm{Connect}(x, y)$ can be represented as reduction: %, where $x, y$ are variables.
  %We have a finite unification reduction path in LP-Unif: 

 {\footnotesize 
  \begin{center}
    $\Phi \vdash \{\mathrm{Connect}(x, y)\} \leadsto_{\kappa_1, [x/x_1, y/z_1]}
    \{\mathrm{Connect}(x, y_1), \mathrm{Connect}(y_1, y)\} \leadsto_{\kappa_2, [\mathrm{node_1}/x, \mathrm{node_2}/y_1, \mathrm{node_1}/x_1, y/z_1]}
    \{\mathrm{Connect}(\mathrm{node_2}, y)\} \leadsto_{\kappa_3, [\mathrm{node_3}/y, \mathrm{node_1}/x, \mathrm{node_2}/y_1,
      \mathrm{node_1}/x_1, \mathrm{node_3}/z_1]} \emptyset $
  \end{center}
  }
  \noindent The first reduction $\leadsto_{\kappa_1, [x/x_1, y/z_1]}$  unifies query $\mathrm{Connect}(x, y)$ with the head of the rule $\kappa_1$, $\mathrm{Connect}(x_1, z_1)$. %, namely, $\mathrm{Connect}(x_1, z_1)$, with $\mathrm{Connect}(x, y)$, the unification generates a substitution $[x/x_1, y/z_1]$.
 Note that $x/x_1$ means $x_1$ is replaced by $x$. After that, the query is \emph{resolved} with the formula of $\kappa_1$, %and  the subtitution is applied to  the body of $\kappa_1$,
  producing the next queries: $\mathrm{Connect}(x, y_1)$, $\mathrm{Connect}(y_1, y)$. %SLD-resolution assumes  renaming variables apart, to avoid incorrect variable capture and circular unification. %We assume it implecitely throughout the paper.  %the free variable in the head if we have a name conflict

\end{example}

Seeing program as Horn clauses, the above derivation first assumed that  $\mathrm{Connect}(x, y)$ is false, and then deduced a contradiction (an empty goal) from the assumption.
%Thus, it followed the general resolution scheme: 
As every SLD-derivation is essentially a proof by contradiction, traditionally, the exact content of such proofs plays little role in determining entailment. Instead, termination of derivations plays a crucial role. When it comes to logical entailment
with respect to programs that admit non-terminating derivations, resolution gives only a semi-decision
procedure. A long-standing challenge has been to find computationally effective
mechanisms that guarantee termination of LP proof search, and to use
them to deduce logical entailment for LP~\cite{deSchreye1994199}.

LP approach of preserving a tight connection between entailment and termination makes it hard to model corecursive computations. There are potentially infinite derivations that may bear some interesting
computational meaning.

\begin{example}\label{ex:str}
  The following program defines the predicate $\mathrm{Stream}$:
 {\footnotesize  
\begin{center}
$\kappa_1 : \forall x . \forall y . \mathrm{Stream}(y) \Rightarrow \mathrm{Stream}(\mathrm{cons}( x, y))$  
\end{center}
}

\noindent It models infinite streams, and will result in infinite derivations, e.g.:

 {\footnotesize
  \begin{center}
    $\Phi \vdash \{\mathrm{Stream}(\mathrm{cons}(x,y))\} \leadsto_{\kappa_1, [x/x_1, y/y_1]}
    \{\mathrm{Stream}(y)\} \leadsto_{\kappa_1, [\mathrm{cons}(x_2, y_2)/y]}
    \{\mathrm{Stream}(y_2)\} \leadsto_{\kappa_1,  [\mathrm{cons}(x_3, y_3)/y_2]} \ldots$
  \end{center}
}

For the query $\mathrm{Stream}(\mathrm{cons}(x,y))$, we may still want to either obtain a description of the solution for the variable $y$, or make finite observation on its solution, however, none of these are supported by standard SLD-resolution. 
\end{example}

Two groups of methods have been proposed to address this problem.

-- CoLP (\cite{Gupta07}, \cite{Simon07}) offers methods for loop invariant analysis in SLD-derivations: infinite derivations are
terminated if a loop of a certain shape is detected in resolvents, e.g., $\mathrm{Stream}(y)$ and $\mathrm{Stream}(y_2)$ above are unifiable, so one may conclude with a regular description $[\mathrm{cons}(x_2, y)/y]$.

-- There are many infinite derivations that do not form a loop, CoALP/S-resolution (\cite{komendantskaya2014}, \cite{JKK15}) aim to provide general coinductive gurantee that for \textit{productive} infinte derivation, one can make finite observation. E.g. in S-resolution, the derivation for $\mathrm{Stream}(\mathrm{cons}(x,y))$ will stop at $\mathrm{Stream}(y_2)$ and report the process is infinite with partial answer $[\mathrm{cons}(x_2, y_2)/y]$, then one can choose to continue the derivation to further inspect $y_2$.

Let us view SLD-derivations as \emph{reductions}, starting from a given query and using \emph{unification} with Horn formulas;  %As such reductions use unification when applying the reduction rules,
%let us
we call such reductions   \emph{LP-Unif reductions}  and denote them by $\leadsto$. If we restrict the unification algorithm underlying such reductions to allow only term-matchers instead of unifiers, we obtain \emph{LP-TM reductions}, denoted by $\to$.
They model computations performed by rewriting trees in \cite{JKK15}, and it has interesting properties distinguishing them from LP-Unif reductions.
Firstly, they may give partial proofs compared to LP-Unif reductions: %, as the following example shows:

\begin{example}\label{ex:bl}
The following program defines bits and lists of bits:
 {\footnotesize
 \begin{center}
     $\kappa_1 :\  \Rightarrow \mathrm{Bit}(0)$\\
   $\kappa_2 :\ \Rightarrow \mathrm{Bit}(1)$\\
    $\kappa_3 :\ \Rightarrow \mathrm{BList}(\mathrm{nil})$\\
  $\kappa_4 : \forall x . \forall y.  \mathrm{BList}(y), \mathrm{Bit}(x) \Rightarrow \mathrm{BList}(\mathrm{cons}( x, y))$\\

\end{center}} 

\noindent  Below is an example of LP-TM reduction  to normal form:
 
 {\footnotesize
  \begin{center} 
    $\Phi \vdash \{\mathrm{BList}(\mathrm{cons}(x,y))\} \to_{\kappa_4}
    \{ \mathrm{Bit}(x),  \mathrm{BList}(y)\}$ 
  \end{center} } 

 %\noindent LP-TM reduction is possible terminates, % without substituting for $x$ or $y$.
% but LP-Unif would give a complete
%proof:

   \noindent %Here, $\to_{\kappa_1, [x/x_1, y/y_1]}$ means: \textit{match} the head of the rule $\kappa_1$ to $\mathrm{Stream}(\mathrm{cons}(x,y) )$.
  %Note that no substitution is applied to the query in the course of the above reduction.
  Above, the head of the rule $\kappa_4$ is matched  to the query by substitution
 $[x/x_1, y/y_1]$, which is applied to the body of the rule $\kappa_4$, thus resolving to
 $\{\mathrm{Bit}(x),\mathrm{BList}(y)\}$. %After that, no further LP-TM reduction is possible, without substituting for $y$. For this program, every LP-TM reduction will be finite.
   
\noindent %LP-TM reductions will only give a partial proof compared to LP-Unif: %, the two kinds of reductions are shown below: %Consider the query $\mathrm{BList}(\mathrm{cons}(x,y))$:
But LP-Unif would be able to complete the proof:
 
 {\footnotesize
  \begin{center} 
    $\Phi \vdash \{\mathrm{BList}(\mathrm{cons}(x,y))\} \leadsto_{\kappa_4, [x/x_1,y/y_1]}
    \{\mathrm{Bit}(x),  \mathrm{BList}(y)\} \leadsto_{\kappa_1, [0/x, 0/x_1,y/y_1]} 
     \{\mathrm{BList}(y)\} \leadsto_{\kappa_3, [\mathrm{nil}/y, 0/x, 0/x_1,\mathrm{nil}/y_1]}  \emptyset$ 
  \end{center} } 

 %, which would require unification.
%\knote{add example?}
\end{example}

On the other hand, LP-TM reductions terminate for programs that are traditionally seen as coinductive:
\begin{example}\label{str2}
  Consider LP-TM reduction for the program of Example~\ref{ex:str}:  
  \begin{center} 
     $\Phi \vdash \{\mathrm{Stream}(\mathrm{cons}(x,y))\} \to_{\kappa_1}
    \{ \mathrm{Stream}(y)\} $ 
  \end{center}  

\end{example}

%Finally, with term-matching reduction, we may have divergent behaviour that was not present in LP-Unif.
Finally, LP-TM reductions are not guaranteed to terminate in general. %, as the following example shows.

\begin{example}\label{ex:con2}

For the program of Example~\ref{ex:conn}, we have the following non-terminating reduction by LP-TM.  
 {\footnotesize

\begin{center} 
  $\Phi \vdash \{\mathrm{Connect}(x, y)\} \to_{\kappa_1} \{\mathrm{Connect}(x, y_1), \mathrm{Connect}(y_1, y)\} $

$\to_{\kappa_1} \{\mathrm{Connect}(x, y_2), \mathrm{Connect}(y_2, y_1), \mathrm{Connect}(y_1, y)\} \to_{\kappa_1} ... $ 
\end{center}
}
%\knote{check with Fu Peng: how do you denote ``x is replaced by a'': $x/a$ or $a/x$? I tend to use $x/a$... In the above example, you do the opposite. At the very least, we must be consistent...}
% However, we can not match the head of $\kappa_2$(namely, $\mathrm{Connect}(\mathrm{node_1}, \mathrm{node_2})$) to neither $\mathrm{Connect}(a, y)$, nor $\mathrm{Connect}(y, b)$. The only rule we can use to proceed the term matching reduction is $\kappa_1$, thus
% the term-matching reduction is not terminating for query $\mathrm{Connect}(a, b)$. 
\end{example}

The programs that admit only finite LP-TM reductions are called productive logic programs (\cite{komendantskaya2014}, \cite{JKK15}). As S-resolution combines LP-TM with unification, finiteness of LP-TM reductions allows one to observe partial answer, while the whole derivation may be infinite. Finiteness of LP-TM is also the key property to ensure this combination of LP-TM with unification is well-behaved, that is, it ensures the operational equivalence of LP-Struct and SLD resolution (we will show this later). 

In  Section~\ref{forms}, we formalise S-resolution as reduction rules that
combine  LP-TM reductions with substitution steps, and call the resulting reductions \emph{LP-Struct reductions}.
%It remains to show exactly how LP-Struct relates to LP-Unif.
%Having observed the three above properties, it is not easy to see how S-resolution compares to LP-Unif.
We see that for the program in Example \ref{ex:conn}, LP-TM reduction will necessarily diverge, 
while for LP-Unif there exists a finite success path. This mismatch between LP-TM and LP-Unif
makes it difficult to establish the operational relation between LP-Unif and LP-Struct. 
As Section~\ref{equiv:unif} shows, they are not operationally equivalent, in general. However, they are equivalent for programs that are productive (have finite TM-reductions)
and non-overlapping (have no common instances for the Horn formula heads).
%\knote{define what it is?}
%The proof of the main result required us to introduce,

In Section~\ref{real:trans}, we introduce a technique called \emph{realizability transformation}, that,
given a program $\Phi$, produces a program $F(\Phi)$ that is productive and non-overlapping. 
%adds an extra argument into each predicate of the given logic program to record proof information.
Realizability transformation is an interesting proof technique on its own, bearing resemblance to Kleene's~\cite{KleeneSC:1952} method under the same name.
Here, it serves several purposes. 1. It helps to define a class of programs for which S-Resolution and SLD-resolution are operationally equivalent.
2. It gives means to record the proof content alongside reductions. 3. It preserves
%; and to prove this equivalence in Section~\ref{equiv:unif}
%But here, it shows a deeper relation between
%our new approach to S-resolution and the structural properties of SLD-resolution. 
%which in fact also serves as a  general method of program transformation
%that
%The realizability transformation guarantees that any transformed program has finite LP-TM reductions; which is equivalent to guaranteeing the productivity of programs in~\cite{JKK15}.
%, this transformation will also guarantee universal productivity of logic programs.
%Moreover, we show that  the transformation  changes neither
proof-theoretic meaning of the original program and computational behaviour of LP-Unif reductions.
%This coincides with Universal 
%It transforms any given logic program into a productive logic program, showing the relation between constructive content of proofs by resolution and productivity property of ~\cite{JKK15}.
%\knote{the above relation better be exaplained in the relevant section, after relevant thm.}--done

%\end{comment}
%\frank{new}

In order to specify the proof-theoretic meaning of various LP-reductions, we introduce a type-theoretic approach to recover the notion of proof in LP.
 It has been noticed  by Girard~\cite{Girard:1989}, that resolution rule  $\frac{A \lor B \ \ \ \neg B \lor D}{A \lor D}$ can be expressed by means of 
the cut rule in  intuitionistic sequent calculus:
$\frac{A \Rightarrow B  \  \  \     B \Rightarrow D}{A \Rightarrow D}$.
Although the resolution rule is classically equivalent to the cut rule, 
the cut rule is better suited for performing computation 
and at the same time preserving constructive content. In Section~~\ref{forms} we devise
a type system reflecting this intuition: if $p_1$ is a proof of  $A \Rightarrow B$ and $p_2$ is a proof of $B \Rightarrow D$, then $\lambda x . p_2 (p_1 x)$ is a proof of $A \Rightarrow D$.  Thus, a proof can be recorded along with each cut rule. The type system we propose gives a proof theoretic interpretation for LP in general, and in particular to S-resolution.
It also allows us to see clearly the proof-theoretic differences between LP-Unif/LP-Struct and LP-TM. Namely, LP-Unif/LP-Struct give proofs for Horn formulas of the form $\forall x . \Rightarrow \sigma A$, while LP-TM gives proofs for $\forall x . \Rightarrow  A$.
%Thus we can use various LP reductions    as proof search procedure that tries to construct  proof for a given query. and~\ref{equiv:unif}
In Sections~\ref{real:trans}, the type system provides a precise tool % to analyze the operational relation between LP-Unif and LP-Struct, i.e. it enables us
to express the realizability transformation and prove it is a meaning-preserving transformation.

Detailed proofs for lemmas and theorems in this paper may be found in the extended version\footnote{Extended version is available from: \url{http://staff.computing.dundee.ac.uk/pengfu/document/papers/tm-lp.pdf}}. 

\section{A Type System for LP: Horn-Formulas as Types}
\label{forms}
 
We first formulate a type system to model LP.  We
%define
%term-matching reduction, unification reduction and substitution reduction
%term-matching reduction, unification reduction, substitution reduction
 show how LP-Unif, LP-TM and LP-Struct can be defined
in terms of reduction rules. We show that LP-Unif and LP-TM are sound with respect to the type system. 
 
\begin{definition}

\

  Term $t \ ::= \ x \ | \ f(t_1,..., t_n)$

  Atomic Formula $A, B, C, D\ ::= \ P(t_1,...,t_n)$

  (Horn) Formula $F \ ::= \ [\forall \underline{x}]. A_1,..., A_n \Rightarrow A$

%  Quantified Formula/LP Formula: $\forall \underline{x} . F$

  Proof Term $p, e \ ::= \ \kappa \ | \ a \ | \ \lambda a . e \ | \ e \ e'$

  Axioms/LP Programs $\Phi \ ::= \cdot \ | \ \kappa :  F, \Phi$
\end{definition}
 
% Term $t$ denotes first order terms. 
Functions of arity zero are called \textit{term constants}, $\mathrm{FV}(t)$ returns all free term variables of $t$.  %The $P$ in $P(t_1,...,t_n)$ is called \textit{predicate}. Predicate of arity zero is called \textit{proposition}. %LP formula $A_1,..., A_n \Rightarrow A$ can be intuitively viewed as $A_1 \to ...\to A_n \to A$, and the order of $A_1,..., A_n$ does not matter.
We use $\underline{A}$ to denote $A_1,..., A_n$, when the number $n$ is unimportant. If $n$ is zero for $\underline{A} \Rightarrow B$, then we write $\Rightarrow B$. Note that $B$ is an atomic formula, but $\Rightarrow B$ is a formula, we distinguish the notion of atomic formulas from (Horn) formulas. The formula $A_1,..., A_n \Rightarrow B$ can be informally read as ``the conjunction of $A_i$ implies $B$''. We write $\forall \underline{x} . F$ for quantifying over all the free term variables in $F$; $[\forall x].F$ denotes $F$ or $\forall x . F$. LP program $B \Leftarrow \underline{A}$ are represented as  $\forall \underline{x} . \underline{A} \Rightarrow B$ and query is an atomic formula. Proof terms are lambda terms, where $\kappa$ denotes a proof term constant and $a$ denotes a proof term variable. 

The following is a new formulation of a type system intended to provide a type theoretic foundation for LP.  
 
\begin{definition}[Horn-Formulas-as-Types System for LP]
  \label{proofsystem}

\
{\footnotesize
\begin{tabular}{lll}
\\
\infer[gen]{e: \forall \underline{x} . F}{e : F}
&

&
\infer[cut]{\lambda \underline{a} . \lambda \underline{b} . (e_2\ \underline{b})\ (e_1\ \underline{a}) : \underline{A}, \underline{B} \Rightarrow C}{e_1 : \underline{A} \Rightarrow D & e_2 : \underline{B}, D \Rightarrow C}
\\
\\
\infer[inst]{e : [\underline{t}/\underline{x}]F}{e : \forall \underline{x} . F}

&

&
\infer[axiom]{\kappa : \forall \underline{x}. F}{(\kappa : \forall \underline{x}. F) \in \Phi}    

  \end{tabular}
}  
\end{definition}
Note that the notion of type is identified with Horn formulas (atomic intuitionistic sequent), not atomic formulas. The usual sequent turnstile $\vdash$ is
internalized as intuitionistic implication $\Rightarrow$. The rule for first order quantification $\forall$  is placed \textit{outside} of the sequent. The cut rule is the only rule that produces new proof terms. In the \textit{cut} rule, $\lambda \underline{a}.t$ denotes $\lambda a_1....\lambda a_n.t $ and $t\ \underline{b}$ denotes $(...(t\ b_1)... b_n)$. The size of $\underline{a}$ is the same as  
$\underline{A}$ and the size of $\underline{b}$ is the same as $\underline{B}$, and $\underline{a}, \underline{b}$ are not free in $e_1, e_2$. 

Our formulation is given in the style of typed lambda calculus and sequent calculus, the intention for this formulation
is to model LP type-theoretically. It has been observed the cut rule and proper axioms in intuitionistic sequent calculus can emulate LP \cite{Girard:1989}(\S 13.4). Here we add a proof term annotation
and make use of explicit quantifiers. Our formulation uses Curry-style in the sense that 
for the \textit{gen} and \textit{inst} rule, we do not modify the structure of the proof terms. Curry-style
formulation allows us to focus on the proof terms generated by applying the \textit{cut} rule. 
 
\begin{definition}[Beta-Reduction]
We define beta-reduction on proof terms as the congruence closure of the following relation:
 $(\lambda a . p) p' \to_\beta [p'/a]p$
\end{definition}

% In fact, beta-reduction is defined as the congruent closure of the above relation. We work on alpha-equivalence of proof terms and
% quantified formulas. % however we do not define any equivalent classes for term in this paper, so $f(x)$
%is not syntactically equal to $f(y)$, i.e. $f(x) \not \equiv f(y)$.

\begin{definition}[Term Matching]
\label{red}
 We define $A \mapsto_{\sigma} B$, $A$ is matchable to $B$ with a substitution $\sigma$ and $t \mapsto_{\sigma} t'$, $t$ is matchable to $t'$ with a substitution $\sigma$. 
\

{\footnotesize
  \begin{tabular}{lllll}
\\
\infer{x \mapsto_{[t/x]} t}{}
&

&
    \infer{P(t_1,..., t_n) \mapsto_{\sigma_1 \cup ... \cup \sigma_n} P(t_1',...,t_n') }{\{t_i \mapsto_{\sigma_i} t_i'\}_{i\in \{1,...,n\}}}
&

    &
\infer{f(t_1,..., t_n) \mapsto_{\sigma_1 \cup ... \cup \sigma_n} f(t_1',...,t_n') }{\{t_i \mapsto_{\sigma_i} t_i'\}_{i\in \{1,...,n\}}}
  \end{tabular}
}
\end{definition}
%% Here we adopt some syntactic conventions to ease the reasoning about term matching. 
% We treat
% substitution as set and the union of them will invoke a syntactic comparison of two term. For example
Here $[t_1/x] \cup [t_2/x] = [t_1/x]$ if $t_1 \equiv t_2$, else, the matching process fails; and $[t_1/x] \cup [t_2/y] = [t_1/x, t_2/y]$. 

%The following is a formulation of Robinson's unification. 

\begin{definition}[Unification]
 We define $A \sim_{\gamma} B$, $A$ is unifiable to $B$ with substitution $\gamma$ and $t \sim_\gamma t'$, $t$ is unifiable with $t'$ with substitution $\gamma$.

\
{\footnotesize
  \begin{tabular}{lllllll}
\\
\infer{x \sim_\emptyset x}{}

&

&
\infer{x \sim_{[t/x]} t}{x \notin \mathrm{FV}(t)}
&

&
\infer{f(t_1,..., t_n) \sim_{\gamma_n \cdot ... \cdot \gamma_1} f(t_1',...,t_n') }{\{\gamma t_i \sim_{\gamma_i} \gamma t_i' \ \ \ \ \gamma := \gamma_{i-1}\cdot...\cdot\gamma_0\}_{i\in \{1,...,n\} } & \gamma_0 = \emptyset}

\\
\\
\infer{t \sim_{[t/x]}x}{x\notin \mathrm{FV}(t)}
&

&

&

&

\infer{P(t_1,..., t_n) \sim_{\gamma_n \cdot ... \cdot \gamma_1} P(t_1',...,t_n') }{\{ \gamma t_i \sim_{\gamma_i} \gamma t_i' \ \ \gamma := \gamma_{i-1}\cdot...\cdot\gamma_0\}_{i\in \{1,...,n\}} & \gamma_0 = \emptyset}
  \end{tabular}
}
\end{definition}
%% The definition of unification above is an imperative definition rather than an inductive one, i.e.
Note that $\gamma$ is updated for each $i$, and $\gamma \cdot \gamma'$ denotes composition of substitutions $\gamma, \gamma'$. %The $x\sim_\emptyset x$ rule always has priority over the other rule. % We state the following well-known lemma for unification (\S 4.5 in \cite{Baader:1998}). 
% \begin{lemma}[Idempotence]
%   If $A \sim_\gamma C$, then $\gamma \cdot \gamma = \gamma$.
% \end{lemma}

% Below, we use  We define different kinds of reduction for LP.

Below, we formulate different notions of reduction for LP. A similar style of formulating SLD-derivation as reduction system appeared in \cite{nilsson1990logic},
but we identify two more kinds of reductions here: term-matching and substitutional reductions.

\begin{definition}[Reductions]
\label{red}
We define reduction relations on the multiset of atomic formulas: 
\begin{itemize}

\item \textbf{Term-matching(LP-TM) reduction:}

$\Phi \vdash \{A_1,..., A_i, ..., A_n\} \to_{\kappa, \gamma'} \{A_1,..., \sigma B_1,..., \sigma B_m, ..., A_n\}$ for any substitution $\gamma'$, if there exists $\kappa : \forall \underline{x} . B_1,..., B_n \Rightarrow C \in \Phi$ such that $C \mapsto_{\sigma} A_i$.

\item \textbf{Unification(LP-Unif) reduction:} 

$\Phi \vdash \{A_1,..., A_i, ..., A_n\} \leadsto_{\kappa, \gamma \cdot \gamma'} \{\gamma A_1,..., \gamma B_1,..., \gamma B_m, ..., \gamma A_n\}$ for any substitution $\gamma'$, if there exists $\kappa : \forall \underline{x} . B_1,..., B_n \Rightarrow C \in \Phi$ such that $C \sim_{\gamma} A_i$.

\item \textbf{Substitutional reduction:} 

$\Phi \vdash \{A_1,..., A_i, ..., A_n\} \hookrightarrow_{\kappa, \gamma \cdot \gamma'} \{\gamma A_1,..., \gamma A_i, ..., \gamma A_n\}$ for any substitution $ \gamma'$, if there exists $\kappa : \forall \underline{x} . B_1,..., B_n \Rightarrow C \in \Phi$ such that $C \sim_{\gamma} A_i$.

\end{itemize}
\end{definition}
The second subscript of term-matching reduction is used to store the substitutions obtained by unification, it is only used when we combine term-matching reductions with substitutional reductions. The second subscript in unification and substitutional reduction is intended as a state, it will be updated along with reductions. If we just talk about term-matching reduction alone, we usually use $\to$ or $\to_\kappa$. We assume implicit renaming of all quantified variables each time the above rule is applied. We write $\leadsto$ and $\hookrightarrow$ when we leave the underlining state implicit. We use $\to^*$ to denote the reflexive and transitive closure of $\to$, similarly for $\leadsto$. Notation $\leadsto_{\gamma}^*$ and $\to_\gamma^*$ is used when the final state along the reduction path is $\gamma$. Notice the
difference between the substitutional reduction and the unification reduction. Unification reduction requires applying the substitution generated by unification to every atomic formula in the multiset. For term-matching reduction, the other atomic formulas are not affected by the computed substitution, thus term-matching reductions can be parallelised.
 
%\begin{definition}[LP-Unif]
%\
%\noindent  LP by unification is defined as $(\Phi, \leadsto)$, namely, given
%a query $B$, LP-Unif uses only unification reduction to reduce $\{B\}$. 
%\end{definition}

Given a program $\Phi$ and
a set of queries $\{B_1, \ldots, B_n\}$, LP-Unif uses only unification reduction to reduce $\{B_1, \ldots, B_n\}$: 

\begin{definition}[LP-Unif]
\
\noindent  Given a logic program $\Phi$, LP-Unif is given by an abstract reduction system $(\Phi, \leadsto)$. 
\end{definition}

Given a program $\Phi$ and
a set of queries $\{B_1, \ldots, B_n\}$, LP-TM uses only term-matching reduction to reduce $\{B_1, \ldots, B_n\}$:

\begin{definition}[LP-TM]
\
\noindent Given a logic program $\Phi$,  LP-TM is given by an abstract reduction system $(\Phi, \to)$. % namely, given
%a query $B$, LP-TM uses only term-matching reduction to reduce $\{B\}$.
\end{definition}
 
LP-TM seems to be a foreign notion for LP, but it is used in \textit{Context Reduction} \cite{Jones97} in type class instance resolution. LP-TM reductions is all we need to define productivity
(\cite{komendantskaya2014}, \cite{JKK15}):

\begin{definition}[Productivity]
  We say a program $\Phi$ is productive iff every $\to$-reduction is finite.
\end{definition}
 
\begin{definition}
  We use $\to^\mu$ to denote a reduction path to a $\to$-normal form. If the
  $\to$-normal form does not exist, i.e. every $\to$-reduction path is infinite, then $\to^\mu$ denotes an infinite reduction path. If we know that $\to$ is strongly normalizing, then we use $\to^\nu$ to denote a reduction path to a $\to$-normal form. We write $\hookrightarrow^1$ to denote at most one step of $\hookrightarrow$.
\end{definition}

%For a productive program is defined using only the term-matching reduction, for a prductive program, it is still possible that for a query to be diverged, due to the present
% of unification(See Example \ref{ex:str}).   

Given a program $\Phi$ and
a set of queries $\{B_1, \ldots, B_n\}$, LP-Struct  first uses term-matching reduction to reduce  $\{B_1, \ldots, B_n\}$ to a normal form, then performs one step substitutional reduction, and then repeats this process.
%to reduce $\{B_1, \ldots, B_n\}$

\begin{definition}[LP-Struct]
\
\noindent  Given a logic program $\Phi$,  LP-Struct is given by an abstract reduction system $(\Phi, \to^{\mu} \cdot \hookrightarrow^1)$.
%LP by structural resolution is defined as $(\Phi, \to^{\mu} \cdot \hookrightarrow^1)$, namely, given a query $B$, LP-Struct

\end{definition}

If a finite term-matching reduction path does not exist, then $\to^{\mu} \cdot \hookrightarrow^1$ denotes an infinite path. 
When we write $\Phi \vdash \{\underline{A}\} (\to^{\mu} \cdot \hookrightarrow^1)^* \{\underline{C}\}$, it means a nontrivial finite path will be of the shape $\Phi \vdash \{\underline{A}\} \to^{\mu} \cdot \hookrightarrow \cdot ... \cdot \to^{\mu} \cdot \hookrightarrow \cdot \to^\mu \{\underline{C}\}$. 
% When $\to$ is not terminating, the LP-Struct
% reduction behaves just as LP-TM. This raises the question of why we are claiming
% LP-Struct and LP-Unif to be equivalent, whereas LP-TM and LP-Unif are known to be different. We will address this question in Section \ref{real:trans}. 

%\subsection{Proof Theoretic Meaning of LP-Unif and LP-TM}
We first show that LP-Unif and LP-TM are sound w.r.t. the type system
of Definition \ref{proofsystem}, which implies that we can obtain a proof for each successful query. 
 
\begin{lemma}
\label{sound}
  If $\Phi \vdash \{A_1,..., A_n\} \leadsto^*_{\gamma} \emptyset$, then there exist proofs $e_1 : \forall \underline{x} . \Rightarrow \gamma A_1,..., e_n : \forall \underline{x} . \Rightarrow \gamma A_n$, given axioms $\Phi$.
\end{lemma}

\begin{proof}
  By induction on the length of the reduction.
  
%\begin{itemize}
%\item
\noindent  \textit{Base Case.} Suppose the length is one, namely, $\Phi \vdash \{A\} \leadsto_{\kappa, \gamma} \emptyset$. It implies that there exists $(\kappa : \forall \underline{x} .  \Rightarrow C) \in \Phi$, such that $C \sim_\gamma A$.  So we have $\kappa :\ \Rightarrow \gamma C$ by the \textit{inst} rule. Thus $\kappa : \ \Rightarrow \gamma A$ by $\gamma C \equiv \gamma A$. Hence $\kappa : \forall \underline{x} . \Rightarrow \gamma A$ by the \textit{gen} rule.

%\item
\noindent  \textit{Step Case.} Suppose
 $\Phi \vdash \{A_1, ..., A_i,..., A_n\} \leadsto_{\kappa, \gamma} \{\gamma A_1,..., \gamma B_1,..., \gamma B_m,..., \gamma A_n\}$ $ \leadsto^*_{\gamma'} \emptyset$, where $\kappa : \forall \underline{x} . B_1,..., B_m \Rightarrow C$ and $C \sim_{\gamma} A_i$. By IH, we know that there exist proofs $e_1 : \forall \underline{x}. \Rightarrow \gamma' \gamma A_1,..., p_1 : \forall \underline{x}. \Rightarrow \gamma' \gamma B_1,..., p_m : \forall \underline{x}. \Rightarrow \gamma' \gamma B_m,..., e_n : \forall \underline{x} . \Rightarrow \gamma' \gamma A_n$. % We
% know that $\gamma' \gamma A = \gamma' A$ for any $A$. So $e_1 : \forall \underline{x}. \Rightarrow \gamma' A_1,..., p_1 : \forall \underline{x}. \Rightarrow \gamma' B_1,..., p_m : \forall \underline{x}. \Rightarrow \gamma' B_m,..., e_n : \forall \underline{x} . \Rightarrow \gamma' A_n$. 
We can use \textit{inst} rule to instantiate the quantifiers of $\kappa$ using $\gamma'\cdot \gamma$, so
we have $\kappa : \gamma'\gamma B_1,..., \gamma'\gamma B_m \Rightarrow \gamma'\gamma C$.
Since $\gamma'\gamma A_i \equiv \gamma' \gamma C$, we can construct a proof $e_i = \kappa \ p_1\ ...\ p_m$ with $e_i :\ \Rightarrow \gamma'\gamma A_i$, by applying the cut rule $m$ times. By \textit{gen}, we have $e_i : \forall \underline{x}. \Rightarrow \gamma'\gamma A_i$. The substitution generated by the unification is idempotent, and $\gamma'$ is accumulated from $\gamma$, i.e. $\gamma' = \gamma'' \cdot \gamma$ for some $\gamma''$,
so $\gamma' \gamma A_j \equiv \gamma'' \gamma \gamma A_j \equiv \gamma''\gamma A_j \equiv \gamma' A_j$ for any $j$. Thus we have $e_j : \forall \underline{x}. \Rightarrow \gamma' A_j$ for any $j$. 
%\end{itemize}
\end{proof}

%Lemma \ref{sound} implies the following theorem. 

\begin{theorem}[Soundness of LP-Unif]
\label{sound:unif}
    If $\Phi \vdash \{A\} \leadsto^*_\gamma \emptyset$ , then there exists a proof $e : \forall \underline{x} .\Rightarrow \gamma A$ given axioms $\Phi$. %% and $e$ is a first order proof term, furthermore, $e$ is in beta normal form.  
\end{theorem}

For example, by the soundness theorem above, the derivation in Example \ref{ex:conn} yields a proof $(\lambda b. (\kappa_1\ b)\ \kappa_3)\ \kappa_2$ for the formula $\ \Rightarrow \mathrm{Connect}(\mathrm{node}_1, \mathrm{node}_3)$. 
%% 
%% \begin{proof}
%%   By lemma \ref{sound}. %% Here we want to mention that since we have a constructive proof, we
%% %% know how exactly to construct such proof $e$. Inspecting the proof of Lemma \ref{sound}, we never use lambda to construct the proof term $e$, thus $e$ is a first order proof term.
%% \end{proof}

\begin{theorem}[Soundness of LP-TM]
\label{sound:tm}
    If $\Phi \vdash \{A\} \to^* \emptyset$ , then there exists a proof $e : \forall \underline{x} .\Rightarrow  A$ given axioms $\Phi$.
\end{theorem} 

%% \begin{proof}
%%   Similar to the proof of Lemma \ref{sound}.
%% \end{proof} 

Observing Theorem \ref{sound:unif} and Theorem \ref{sound:tm}, we see that for LP-TM, there is no need to accumulate substitutions, and the resulting formula is proven as stated, and does not require substitution. This difference is due to the difference of LP-TM and LP-Unif reductions. We are going to
postpone the proof of soundness theorem for LP-Struct to Section \ref{equiv:unif}, there we
show LP-Struct and LP-Unif are operationally equivalent, which implies the soundness of LP-Struct. 

\begin{comment}
However, LP-TM is incomplete w.r.t. the proof 
system, as we will see. 

\begin{lemma}
\label{comp:l1}
If there exists proofs $e : [\forall \underline{x} .] B_1,..., B_m \Rightarrow  A$ given axioms $\Phi$, then for any $\delta$, $\Phi \vdash \{\delta A\} \leadsto^*_{\gamma} \{\gamma_1 B_1,..., \gamma_m B_m\}$ for some $\gamma,\gamma_1,..., \gamma_m$.
\end{lemma}
\begin{proof}
  By induction on derivation of $e : [\forall \underline{x} .] B_1,..., B_m \Rightarrow  A$.
\end{proof}
 
\begin{theorem}[Completeness of LP-Unif]
\label{comp:unif}
If there exists proofs $e : [\forall \underline{x}] .\Rightarrow  A$ given axioms $\Phi$, then  $\Phi \vdash \{\delta A\} \leadsto^* \emptyset$ for any $\delta$. 
\end{theorem} 
\begin{proof}
  By Lemma \ref{comp:l1}. 
\end{proof}
\textbf{Incompleteness of LP-TM} For logic program $\kappa_1 : \forall x.\forall y . T(y) \Rightarrow T(x), \kappa_2 : \ \Rightarrow T(c)$ where $c$ is a term constant. We know that there exist a proof $\kappa_1 \kappa_2 : \forall x . \Rightarrow T(x)$, but $\Phi \vdash \{ T(x)\} \to \{ T(y)\} \not \to $, since $T(c)$ is not matchable to $T(y)$. 
\end{comment}

 % We are going to pursue approach 2. in this paper. To overcome the difficulty of having $\to$ to be diverging, we are going to transform $\to^\mu$ to $\to^\nu$, and the infinite behavior of $\to^\mu$ can be exibited by the infinite behavior of $(\to^{\nu} \cdot \hookrightarrow^1)$. We will see this more clearly in the following Section.

\section{Realizability Transformation}
\label{real:trans} 
We define \textit{realizability transformation} in this section. Realizability described in \cite{KleeneSC:1952}(\S 82) is a technique that uses a number to represent a proof of a number-theoretic formula. The transformation described here is similar in the sense that we use a first order term to represent the proof of a formula. 
More specifically, we use a first order term as an extra argument for a formula to represent a proof of that formula. Before we define the transformation, we first state several basic results about
the type system in Definition \ref{proofsystem}. 
 
\begin{theorem}[Strong Normalization]
\label{real:sn}
  If $e : F$, then $e$ is strongly normalizable w.r.t. beta-reduction on proof terms.
\end{theorem}
 
The proof of strong normalization (SN) is an adaptation of Tait-Girard's reducibility proof. Since the first order quantification does not impact the proof term, the proof is very similar to the SN proof of simply typed lambda calculus. 
 
\begin{lemma}
\label{fst:lambda}
   If $e : [\forall \underline{x}.] \underline{A} \Rightarrow B$ given axioms $\Phi$, then either $e$ is a proof term constant or it is normalizable to the form $\lambda \underline{a}. n$, where $n$ is first order normal proof term. 
\end{lemma}
 
\begin{theorem}
  \label{fst}
    If $e : [\forall \underline{x}.] \Rightarrow B$, then $e$ is normalizable to a first order proof term.
\end{theorem}
 
Lemma \ref{fst:lambda} and Theorem \ref{fst} show that  we can use first order terms to represent normalized proof terms; and thus pave the way to realizability transformation.
 
\begin{definition}[Representing First Order Proof Terms]
\label{fst:rep}
  Let $\phi$ be a mapping from proof term variables to first order terms. We define 
a representation function $\interp{\cdot}_\phi$ from first order normal proof terms to first order terms. 
  
\noindent -- $\interp{a}_\phi = \phi(a)$. 
%  \item $\interp{\lambda a_1 ... a_n . p}_{\phi} = \interp{p}_{\phi[\alpha_1/a_1,..., \alpha_n/a_n]}$.

  \noindent -- $\interp{\kappa \ p_1 ...p_n}_\phi = f_{\kappa}(\interp{p_1}_\phi,..., \interp{p_n}_\phi)$, where $f_\kappa$ is a function symbol.

\end{definition}
 
\begin{definition}
  Let $A \equiv P(t_1,..., t_n)$ be an atomic formula, we write
  $A[t']$, where $(\bigcup_i \mathrm{FV}(t_i)) \cap \mathrm{FV}(t') = \emptyset$, to 
  abbreviate a new atomic formula $P(t_1,..., t_n, t')$.
\end{definition}
 
% We define the function $\interp{\cdot}_\phi$ to represent
% \textit{flat}\footnote{It is a special form of supercombinator, where the only lambda binders are at top-level.} lambda proof terms as first order terms. %We can see this in the following definition.
\begin{definition}[Realizability Transformation]
\label{real}
  We define a transformation $F$ on formula and its normalized proof term: 
  \begin{itemize}
  \item $F(\kappa : \forall \underline{x} . A_1, ..., A_m \Rightarrow B) = \kappa : \forall \underline{x} . \forall \underline{y}. A_1[y_1], ..., A_m[y_m] \Rightarrow B[f_\kappa(y_1,...,y_m)]$, where $y_1,..., y_m$ are all fresh and distinct.
  \item $F(\lambda \underline{a} . n : [\forall \underline{x}] . A_1, ..., A_m \Rightarrow B) = \lambda \underline{a} . n : [\forall \underline{x}.\forall \underline{y}]. A_1[y_1], ..., A_m[y_m] \Rightarrow $ 

\noindent $B[\interp{n}_{[\underline{y}/\underline{a}]}]$, where $y_1,..., y_m$ are all fresh and distinct.
  \end{itemize}
     
\end{definition}

The realizability transformation systematically associates a proof to each atomic formula,
so that the proof can be recorded along with reductions. % The strong normalization result
% is crutial to realizability transformation, since normalized proof can be represented by a first-order term. 

% The following theorems show the realizability transformation has three highly desirable properties, namely,
% the transformation does not change the proof theoretic meanings for LP-Unif, it yields a finite term-matching reduction for LP-TM, and it allows us to compute the proof term for the goal formula for LP-Unif. 
%% The following theorem shows the realizability transformation yields a finite term-matching reduction for LP-TM. 
\begin{example}
 \label{ex:conn:real0}
The following logic program is the result of applying realizability transformation on
the program in Example \ref{ex:conn}.
  {\footnotesize
  \begin{center}
  $  \kappa_1 : \forall x . \forall y . \forall u_1. \forall u_2 . \mathrm{Connect}(x, y, u_1), \mathrm{Connect}(y, z, u_2) \Rightarrow \mathrm{Connect}(x, z, f_{\kappa_1}(u_1, u_2))$

  $\kappa_2 : \ \Rightarrow \mathrm{Connect}(\mathrm{node_1}, \mathrm{node_2}, c_{\kappa_2})$
  
  $\kappa_3 : \ \Rightarrow \mathrm{Connect}(\mathrm{node_2}, \mathrm{node_3}, c_{\kappa_3})$
    \end{center}
 }

\noindent Before the realizability transformation, we have the following judgement:
{\footnotesize
\begin{center}
  $\lambda b. (\kappa_1\ b)\ \kappa_2 :
  \mathrm{Connect}(\mathrm{node_2}, z) \Rightarrow
  \mathrm{Connect}(\mathrm{node_1}, z)$
\end{center}
}
\noindent We can apply the transformation, we get: 
{\footnotesize
\begin{center}
  $\lambda b. (\kappa_1\ b)\ \kappa_2 :
  \mathrm{Connect}(\mathrm{node_2}, z, u_1) \Rightarrow
  \mathrm{Connect}(\mathrm{node_1}, z, \interp{(\kappa_1\ b)\ \kappa_2}_{[u_1/b]})$
\end{center}
}
\noindent which is the same as
{\footnotesize
\begin{center}
  $\lambda b. (\kappa_1\ b)\ \kappa_2 :
  \mathrm{Connect}(\mathrm{node_2}, z, u_1) \Rightarrow
  \mathrm{Connect}(\mathrm{node_1}, z, f_{\kappa_1}( u_1, c_{\kappa_2}))$
\end{center}
}

\noindent Observe that the transformed formula:

\noindent $\mathrm{Connect}(\mathrm{node_2}, z, u_1) \Rightarrow \mathrm{Connect}(\mathrm{node_1}, z, f_{\kappa_1}( u_1, c_{\kappa_2}))$ is provable by $\lambda b. (\kappa_1\ b)\ \kappa_2$ using the transformed program.
\end{example}

Let $F(\Phi)$ mean applying the realizability transformation to every axiom in $\Phi$. We
write $(F(\Phi), \leadsto), (F(\Phi), \to), (F(\Phi), \to^\mu \cdot \hookrightarrow^1)$, to
mean given axioms $F(\Phi)$, use LP-Unif, LP-TM, LP-Struct respectively to reduce a given query. Note that for query $A$ in $(\Phi, \leadsto), (\Phi, \to), (\Phi, \to^\mu \cdot \hookrightarrow^1)$, it becomes query $A[t]$ for some $t$ such that $\mathrm{FV}(A) \cap \mathrm{FV}(t) = \emptyset$ in $(F(\Phi), \leadsto), (F(\Phi), \to), (F(\Phi), \to^\mu \cdot \hookrightarrow^1)$.

The next Theorem establishes that, for any program $\Phi$, LP-TM reductions for $F(\Phi)$ are strongly normalizing.

\begin{theorem}
\label{terminate}
 For any $(\Phi, \to^\mu \cdot \hookrightarrow^1)$, we have $(F(\Phi), \to^\nu \cdot \hookrightarrow^1 )$. % i.e. $\to$ is strongly normalizing in $(F(\Phi), \to)$. 
\end{theorem} 
\begin{proof}
 We just need to show $\to$-reduction is strongly normalizing in $(F(\Phi), \to)$. By Definition \ref{fst:rep} and \ref{real}, we can establish a decreasing measurement(from right to left) for each rule in $F(\Phi)$, since
the last argument in the head of each rule is strictly larger than the ones in the body. 
\end{proof}

The above theorem shows that we can use realizability transformation to obtain productive logic programs, moreover, this transformation is general, meaning that any logic program can be transformed to an equivalent productive one. The following theorem shows that realizability transformation does not change the proof-theoretic meaning of a program.

\begin{theorem}\label{th6}
\label{realI}
Given axioms $\Phi$, if $e: [\forall \underline{x}] . \underline{A}\Rightarrow B$ holds with $e$ in normal form, then $F(e : [\forall \underline{x}] . \underline{A}\Rightarrow B)$ holds for axioms $F(\Phi)$.
\end{theorem}
 
% \begin{lemma}
%     If $F(\Phi) \vdash \{A_1[y_1],..., A_n[y_n]\} \leadsto^*_{\gamma} \emptyset$, where $y_1,..., y_n$ are fresh, then there exists proofs $e_1 : \forall \underline{x} . \Rightarrow \gamma A_1[\gamma y_1],..., e_n : \forall \underline{x} . \Rightarrow \gamma A_n[\gamma y_n]$ with $\interp{e_i}_{\emptyset} = \gamma y_i $ given axioms $F(\Phi)$.  
%\end{lemma}
  The other direction for the theorem above is not true if we ignore the transformation $F$, namely, if $e : \forall \underline{x} . \Rightarrow A[t] $ for axioms $\Phi$, it may not be the case
that $e: \forall \underline{x} . \Rightarrow A$, since the axioms $\Phi$ may not be set up in a way
such that $t$ is a representation of proof $e$. The following theorem shows that the extra argument is used to record the term representation of the corresponding proof.
 
\begin{theorem}\label{th7}
\label{record}
 Suppose $F(\Phi) \vdash \{A[y]\} \leadsto^*_{\gamma} \emptyset$. We have $p : \forall \underline{x} . \Rightarrow \gamma A[\gamma y]$ for $F(\Phi)$, where $p$ is in normal form and $\interp{p}_{\emptyset} = \gamma y$. 
\end{theorem} 

Now we are able to show that realizability transformation will not change the unification reduction behaviour.
 
\begin{lemma}
\label{sc:unif}
  $\Phi \vdash \{A_1,..., A_n\} \leadsto^* \emptyset$ iff $F(\Phi) \vdash \{A_1[y_1],..., A_n[y_n]\} \leadsto^* \emptyset$. 
\end{lemma} 
\begin{proof}
  For each direction, by induction on the length of the reduction. Each proof will be similar to the proof of Lemma \ref{sound}, see the extended version for the details. 
\end{proof} 
\begin{theorem}\label{th8}
\label{preservation}
  $\Phi \vdash \{A\} \leadsto^* \emptyset$ iff $F(\Phi) \vdash \{A[y]\} \leadsto^* \emptyset$. 
\end{theorem}
%% \begin{proof}
%%   By Lemma \ref{sc:unif}.
%% \end{proof}

\begin{example}
 \label{ex:conn:real}
Consider the logic program after realizability transformation in Example \ref{ex:conn:real0}. Realizability transformation does not change the behaviour of LP-Unif, we still have the 
  following successful unification reduction path for query $\mathrm{Connect}(x, y, u)$:
   
  \begin{center}
{\footnotesize
$F(\Phi) \vdash \{\mathrm{Connect}(x, y, u)\}\leadsto_{\kappa_1, [x/x_1, y/z_1, f_{\kappa_1}(u_3, u_4)/u]} \{\mathrm{Connect}(x, y_1, u_3), \mathrm{Connect}(y_1, y, u_4)\}$

$\leadsto_{\kappa_2, [c_{\kappa_2}/u_3,\mathrm{node_1}/x, \mathrm{node_2}/y_1, \mathrm{node_1}/x_1, b/z_1, f_{\kappa_1}(c_{\kappa_2}, u_4)/u]} $

$\{\mathrm{Connect}(\mathrm{node_2}, y, u_4)\}$

$\leadsto_{\kappa_3, [c_{\kappa_3}/u_4, c_{\kappa_2}/u_3, \mathrm{node_3}/y, \mathrm{node_1}/x, \mathrm{node_2}/y_1,\mathrm{node_1}/x_1, \mathrm{node_3}/z_1, f_{\kappa_1}(c_{\kappa_2}, c_{\kappa_3})/u]} \emptyset $
}
% {\footnotesize
% $F(\Phi) \vdash \{\mathrm{Connect}(x, y, u)\}\leadsto_{\kappa_1, [x/x_1, y/z_1, f_{\kappa_1}(u_3, u_4)/u]} \{\mathrm{Connect}(x, y_1, u_3), \mathrm{Connect}(y_1, y, u_4)\}\leadsto_{\kappa_2, [c_{\kappa_2}/u_3,\mathrm{node_1}/x, \mathrm{node_2}/y_1, \mathrm{node_1}/x_1, b/z_1, f_{\kappa_1}(c_{\kappa_2}, u_4)/u]} \{\mathrm{Connect}(\mathrm{node_2}, y, u_4)\} \leadsto_{\kappa_3, [c_{\kappa_3}/u_4, c_{\kappa_2}/u_3, \mathrm{node_3}/y, \mathrm{node_1}/x, \mathrm{node_2}/y_1,\mathrm{node_1}/x_1, \mathrm{node_3}/z_1, f_{\kappa_1}(c_{\kappa_2}, c_{\kappa_3})/u]} \emptyset $
% }
\end{center}
 
\end{example}
%\subsection{About Terminating Measurement}

The realizability transformation uses the extra argument as decreasing measurement
to the program to achieve the termination of $\to$-reduction. We want to point out that realizability 
transformation does not modify the proof-theoretic meaning and the execution behaviour. 
The next example shows that not every transformation  technique for obtaining productive programs have such properties:

\begin{example} Consider the following program:

{\footnotesize
\begin{center}
\noindent $\kappa_1 : \ \Rightarrow P(\mathrm{int})$  

\noindent $\kappa_2 : \forall x . P(x), P(\mathrm{list}(x)) \Rightarrow P(\mathrm{list}(x))$
\end{center} 
}
\noindent It is a folklore method to add a structurally decreasing argument as a  measurement to ensure finiteness of $\to^\mu$.

{\footnotesize
\begin{center}

\noindent $\kappa_1 : \ \Rightarrow P(\mathrm{int}, 0)$  

\noindent $\kappa_2 : \forall x . \forall y . P(x, y), P(\mathrm{list}(x), y) \Rightarrow P(\mathrm{list}(x), \mathrm{s}(y))$
\end{center} 
} 
\noindent We denote the above program as $\Phi'$. Indeed with the measurement we add, the term-matching reduction in $\Phi'$ will be finite. But the reduction for query $P(\mathrm{list}(\mathrm{int}), z)$ using unification will fail: 

{\footnotesize
\begin{center}
$\Phi' \vdash \{P(\mathrm{list}(\mathrm{int}), z) \}\leadsto_{\kappa_2, [ \mathrm{int}/x,\mathrm{s}(y_1)/z]} \{P(\mathrm{int}, y_1), P(\mathrm{list}(\mathrm{int}), y_1)\}\leadsto_{\kappa_2, [ 0/y_1, \mathrm{int}/x,\mathrm{s}(0)/z]} \{P(\mathrm{list}(\mathrm{int}), 0)\} \not \leadsto$ 
\end{center}
}
\noindent However, the query $P(\mathrm{list}(\mathrm{int}))$ on the original program using unification reduction will diverge. Divergence and failure are operationally different. Thus
adding arbitrary measurement may modify the execution behaviour of a program (and hence the meaning of the program), but
by Theorems~\ref{th6}-\ref{th8},
realizability transformation does not modify the execution behaviour of unification reduction.
\end{example}

\section{Operational Equivalence of LP-Struct and LP-Unif}
\label{equiv:unif}

%By Theorem \ref{realI} and \ref{preservation}, we know that
Since realizability transformation does not change the proof theoretic meaning of the program or modify the behaviour of unification reduction,  we will work directly on $F(\Phi)$ in this section.
We will show that LP-Struct and LP-Unif 
are equivalent after the realizability transformation. By Theorem \ref{terminate}, it suffices to consider $(F(\Phi), \to^{\nu} \cdot \hookrightarrow^1)$ for LP-Struct. 

The following lemma shows that each LP-Unif reduction can be emulated by one step of substitutional reduction followed by one step of term-matching reduction. 
 
\begin{lemma}
\label{unif-coalp}
  If $F(\Phi) \vdash \{A_1,..., A_i , ..., A_n\} \leadsto_\gamma \{\gamma A_1,..., \gamma \underline{B},..., \gamma A_n\}$ for $\kappa : \forall \underline{x} . \underline{B} \Rightarrow C \in F(\Phi)$ such that $C \sim_{\gamma} A_i$, then $F(\Phi) \vdash \{A_1,..., A_i , ..., A_n\}$ $ \hookrightarrow_{\kappa, \gamma} \{\gamma A_1,..., \gamma A_i , ..., \gamma A_n\} \to_{\kappa}  \{\gamma A_1,..., \gamma \underline{B} ,..., \gamma A_n\}$. 
\end{lemma}

The following lemma shows that for $\to$-normal form, each $\hookrightarrow \cdot \to$ step  is equivalent to a step of $\leadsto$ reduction.
\begin{lemma}
\label{coalp-unif}
 Let $\{A_1[x_1], ..., A_n[x_n]\}$ be a multiset of atomic formulas in $\to$-normal form, and suppose there exists 

\noindent $\kappa : \forall \underline{x}. \underline{y} . B_1[y_1], ..., B_m[y_m] \Rightarrow C[f_\kappa(y_1,..., y_m)] \in F(\Phi)$ such that 

\noindent $C[f_\kappa(y_1,..., y_m)] \sim_{\gamma} A_i[x_i]$. Then we have the following: 
  
\begin{enumerate}
\item $F(\Phi) \vdash \{A_1[x_1],..., A_i[x_i],..., A_n[x_n]\} \hookrightarrow_{\kappa, \gamma}  \{\gamma A_1[x_1],..., \gamma A_i[\gamma x_i], ...,\gamma A_n[x_n]\}$

\noindent $\to_{\kappa}\{\gamma A_1,..., \gamma B_1[y_1],..., \gamma B_m[y_m],...,\gamma A_n[y_n] \}$, 

\noindent with $\{\gamma A_1[x_1],..., \gamma B_1[y_1],..., \gamma B_m[y_m],...,\gamma A_n[x_n] \}$ in $\to$-normal form.

\item $F(\Phi) \vdash \{A_1[x_1],..., A_i[x_i],..., A_n[x_n]\} \leadsto_{\kappa, \gamma}$

\noindent $\{\gamma A_1[x_1],..., \gamma B_1[y_1],..., \gamma B_m[y_m],...,\gamma A_n[x_n] \}$.
\end{enumerate}
 
\end{lemma}

\begin{proof} We only prove \textit{1.} here. 
  We know $\{y_1,..., y_m, x_1, ..., x_{i-1}, x_{i+1}, .., x_n\}\cap
    \mathrm{dom}(\gamma) = \emptyset$, $x_i \in \mathrm{dom}(\gamma)$, and every head in
    $F(\Phi)$ is of the form $D[f(z)]$, so $\{\gamma A_1[x_1],..., \gamma B_1[y_1],..., \gamma
    B_m[y_m],...,\gamma A_n[x_n] \}$ is in $\to$-normal form.
  
\end{proof} 
\begin{theorem}[Equivalence of LP-Struct and LP-Unif]
\label{equiv}
  $F(\Phi) \vdash \{A[y]\} \leadsto^* \emptyset$ iff $F(\Phi) \vdash \{A[y]\} (\to^\nu \cdot \hookrightarrow^1)^* \emptyset$.
\end{theorem} 
\begin{proof}
  From left to right, by Lemma \ref{unif-coalp} and Lemma \ref{coalp-unif}(1), we know that each $\leadsto$
step can be simulated by $\hookrightarrow \cdot \to$. From right to left, by Lemma \ref{coalp-unif}(1), 
we know that the concrete shape of $F(\Phi) \vdash \{A[y]\} (\to^\nu \cdot \hookrightarrow^1)^* \emptyset$ must be of the form $F(\Phi) \vdash \{A[y]\} (\hookrightarrow \cdot \to)^* \emptyset$, then by Lemma \ref{coalp-unif}(2), we have $F(\Phi) \vdash \{A[y]\} \leadsto^* \emptyset$.
\end{proof}

\begin{example}
\label{ex:conn:real1}
For the program in Example \ref{ex:conn:real0}, the query $\mathrm{Connect}(x, y, u)$ can be reduced by LP-Struct successfully:   
 
  \begin{center}
{\footnotesize  
$F(\Phi) \vdash \{\mathrm{Connect}(x, y, u)\} \hookrightarrow_{\kappa_1, [x/x_1, y/z_1, f_{\kappa_1}(u_3, u_4)/u]} \{\mathrm{Connect}(x, y,f_{\kappa_1}(u_3, u_4))\} \to_{\kappa_1} \{\mathrm{Connect}(x, y_1, u_3), \mathrm{Connect}(y_1, y, u_4)\}$

$\hookrightarrow_{\kappa_2, [c_{\kappa_2}/u_3,\mathrm{node_1}/x, \mathrm{node_2}/y_1, \mathrm{node_1}/x_1, b/z_1, f_{\kappa_1}(c_{\kappa_2}, u_4)/u]} \{\mathrm{Connect}(\mathrm{node_1},\mathrm{node_2}, c_{\kappa_2}), \mathrm{Connect}(\mathrm{node_2}, y, u_4)\} \to_{\kappa_2} \{\mathrm{Connect}(\mathrm{node_2}, y, u_4)\}$

$\hookrightarrow_{\kappa_3, [c_{\kappa_3}/u_4, c_{\kappa_2}/u_3, \mathrm{node_3}/y, \mathrm{node_1}/x, \mathrm{node_2}/y_1,\mathrm{node_1}/x_1, \mathrm{node_3}/z_1, f_{\kappa_1}(c_{\kappa_2},c_{\kappa_3})/u]}  \{\mathrm{Connect}(\mathrm{node_2}, \mathrm{node_3}, c_{\kappa_3})\}  \to_{\kappa_3} \emptyset $}
\end{center}
   
\noindent Note that the answer for $u$ is $f_{\kappa_1}(c_{\kappa_2},c_{\kappa_3})$, which is the first order term representation of the proof of $ \ \Rightarrow \mathrm{Connect}(\mathrm{node}_1, \mathrm{node}_3)$. 
\end{example}

After the realizability transformation, LP-Struct is equivalent to LP-Unif in the sense of Theorem \ref{equiv}. As a consequence, we have the soundness theorem for LP-Struct w.r.t. the type system in Definition \ref{proofsystem}. 
\begin{corollary}[Soundness of LP-Struct]
  If $F(\Phi) \vdash \{A[y]\} (\to^\nu \cdot \hookrightarrow^1)_\gamma^* \emptyset$, then
  there exist $e : \forall \underline{x}. \ \Rightarrow \gamma (A[y])$ for $F(\Phi)$.
\end{corollary}
% \begin{proof}
%   We show that LP-Unif is sound, and by Theorem \ref{equiv}, we know LP-Struct is equivalent to LP-Unif, thus
% we can conclude that LP-Struct is sound.
% \end{proof}
%\section{Non-overlapping Condition for LP-Struct and LP-Unif}

We have seen that without realizability transformation, LP-Struct is not operationally equivalent to LP-Unif by Example \ref{ex:conn}.
Example \ref{ex:conn:real1} shows that after realizability transformation, we do get operational equivalence of LP-Struct and LP-Unif. The
mismatch of LP-Unif and LP-Struct seems to be due to the infinity of the $\to$-reduction. One
may wonder whether it is the case that for any productive program,
%(i.e. contains only finite $\to$-reduction path),
%if it is the case that
LP-Struct and LP-Unif are operationally equivalent. The following example shows that 
it is not the case in general.

\begin{example}
\label{overlap}
\
{\footnotesize
  \begin{center}

    \noindent $\kappa_1 :\ \Rightarrow P(\mathrm{c})$ 

   \noindent  $\kappa_2 :  \forall x . Q(x) \Rightarrow P(x)$
  \end{center}
}
\noindent Here $\mathrm{c}$ is a constant. The program is $\to$-terminating. However, for query $P(x)$, we have $\Phi \vdash \{P(x)\} \leadsto_{\kappa_1, [\mathrm{c}/x]} \emptyset$ with LP-Unif, but $\Phi \vdash \{P(x)\} \to_{\kappa_2} \{Q(x)\} \not \hookrightarrow$ for LP-Struct.
\end{example}

So termination of $\to$-reduction is insufficient for establishing the relation between LP-Struct and LP-Unif. In Example \ref{overlap}, the problem is caused by
the overlapping heads $P(\mathrm{c})$ and $P(x)$. Motivated by the notion of non-overlapping in term rewrite system (\cite{bezem2003term}, \cite{Baader:1998}), we have the following definition.

\begin{definition}[Non-overlapping Condition]
  Axioms $\Phi$ are non-overlapping if for any two formulas $\forall \underline{x}. \underline{B} \Rightarrow C, \forall \underline{x}. \underline{D} \Rightarrow E \in \Phi$, there are no substitution $\sigma, \delta$ such that $\sigma C \equiv \delta C'$.
\end{definition}

The theorem below shows that for any non-overlapping program, terminating reductions in LP-Struct are operationally equivalent to terminating reductions in LP-Unif.
However, without productivity, LP-Struct and LP-Unif are no longer operationally equivalent for the diverging program. 

% \noindent As a consequence of non-overlapping condition, we have the following
% theorem. 
\begin{theorem}
  \label{ortho:equiv}
  Suppose $\Phi$ is non-overlapping. $\Phi \vdash \{A_1, ..., A_n\} \leadsto^*_\gamma \{C_1, ..., C_m\}$ with $\{C_1, ..., C_m\}$ in $\leadsto$-normal form iff $\Phi \vdash \{A_1, ..., A_n\} (\to^\mu \cdot \hookrightarrow^1)^*_\gamma \{C_1, ..., C_m\}$ with $\{C_1, ..., C_m\}$ in $\to^\mu \cdot \hookrightarrow^1$-normal form.   
  %% \item  $\Phi \vdash \{A\}  (\to^\mu \cdot \hookrightarrow^1)^*_\gamma  \emptyset$ iff $\Phi \vdash \{A\} \leadsto^*_\gamma \emptyset $. 
  %% \item $\{A\}$ is not $\leadsto$-terminating iff $\{A\}$ is not $\to^\mu\cdot \hookrightarrow$-terminating. 
%  \end{itemize}
\end{theorem}

\begin{example}
Consider the following non-productive and non-overlapping program and its version after the realizability transformation:

\begin{center}
\emph{Original program:}  $\kappa : \forall x . P(x) \Rightarrow P(x)$\\
%\end{center}
%\begin{center}
\emph{After transformation:}  $\kappa : \forall x . \forall u. P(x, u) \Rightarrow P(x, f_\kappa(u))$
\end{center}

\noindent Both LP-Struct and LP-Unif will diverge for the queries $P(x), P(x, y)$ in both original and transformed versions. LP-Struct reduction diverges for different reasons in the two cases, one is due to divergence of
$\to$-reduction:

\noindent $\Phi \vdash \{P(x)\} \to \{P(x)\} \to \{P(x)\} ...$

\noindent The another is due to $\hookrightarrow$-reduction:

\noindent $\Phi \vdash \{P(x, y)\} \hookrightarrow \{P(x, f_k(u))\} \to \{P(x, u)\} \hookrightarrow \{P(x, f_k(u'))\} \to \{P(x, u')\} ...$

 Note that a single step of LP-Unif reduction for the original program corresponds to infinite steps of term-matching reduction in LP-Struct. For the transformed version, a single step of LP-Unif reduction corresponds to finite steps of LP-Struct reduction. 

% It is the ability to distinguish the productive program from the non-productive one that makes LP-Struct more refined than LP-Unif. 

%We can apply realizability transformation
%to this program, we obtain a productive and non-overlapping program:

%\noindent For this program, again, both LP-Struct and LP-Unif will diverge, but the reason
%for the divergence of LP-Struct is now due to the presence of substitutional reduction.

%Observe the similarity between the transfomed program and the program defining Stream in Example~\ref{str}. 
\end{example}

The next theorem shows that we need both productivity and non-overlapping to establish operational
equivalence of LP-Struct and LP-Unif for both finite and infinite reductions. Note that realizability transformation guarantees exactly these two properties.

%But LP-Struct gives us the opportunity to distinguish the productive program from the unproductive ones. Now let us consider the non-overlapping productive program, for such programs, we have the following proposition.
\begin{theorem}\label{prod-non-overlap}
  \
  Suppose $\Phi$ is non-overlapping and productive. 
  \begin{enumerate}
  \item If $\Phi \vdash \{A_1,..., A_n\} \leadsto \{B_1,..., B_m\}$, then $\Phi \vdash \{A_1,..., A_n\} (\to^\nu \cdot \hookrightarrow^1)^* \{C_1,..., C_l\}$ and $\Phi \vdash \{B_1,..., B_m\} \to^* \{C_1,..., C_l\}$.
   \item If $\Phi \vdash \{A_1,..., A_n\} (\to^\nu \cdot \hookrightarrow^1)^* \{B_1,..., B_m\}$, then $\Phi \vdash \{A_1,..., A_n\} \leadsto^* \{B_1,..., B_m\}$.  
  \end{enumerate}
\end{theorem}

%So for the productive and non-overlapping program, LP-Struct behaves almost the same as LP-Unif.
%The only difference of the two can be seen at Proposition \ref{prod-non-overlap} (\textit{1}), 
%that is, LP-Struct may over reduce the queries. But this will not be a problem for non-overlapping program, because this over-reduction can be compensated by more LP-Unif reduction. An example of non-overlapping and productive program can be found at Example \ref{ex:bl}. For the program $F(\Phi)$(after realizability transformation), it is productive and non-overlapping, further more, LP-Struct behaves exactly the same as LP-Unif. 

%In general, for productive and non-overlapping programs, adopting LP-Struct means that the program 

For the diverging but productive programs (like Stream of Example~\ref{ex:str}), 
 productivity gives opportunity to make finite observations for potentially infinite derivations \cite{JKK15}, and allows us not to eargerly unfold the infinite derivation.

% LP-Struct allows a more interactive approach to the diverging program: finite term-matching reduction can ``terminate early '' for a query and gives more options on deciding whether or not to proceed the computation.  

%% If we look back at the Example \ref{ex:conn}, indeed the diverging LP-TM reduction is caused
%% by overlapping between $\mathrm{Connect}(x, z)$ and $\mathrm{Connect}(\mathrm{node}_1, \mathrm{node}_2)$. Intuitively, overlapping Horn formulas exibit the differences between unification reduction and term-matching reduction, in the present of overlapping, for unification reduction it means
%% choosing a rule nondeterministically, so there is still a chance to choose for unfication reduction, but for term-matching reduction, sometimes there is no way to choose, just like Example \ref{ex:conn} and Example \ref{overlap}. Requiring non-overlapping eliminates the flexibility of unification reduction, resulting the equivalence between LP-Struct and LP-Unif. 

\section{Conclusions and Future Work}\label{concl}

%\knote{place here:?}

%Thus, the paper's contributions are two-fold.
  
%\frank{Overall, I agree we should find a balance in mentioning structure resolution(and it is not easy!). And I think lics
%paper is an important paper. I guess the challenging part is how to explain the intuition and contribution of this paper
%in a way that the reader doesn't need to go to read the lics paper.} -- yes, see if I succeeded?
%Proof structure is not so relevant in the SLD-proofs.%
%Even more 
%Thus, resolution is in essence a proof by contradiction, in which the intuitionistic ethos is not inherently found.

%\noindent \textbf{About the Proof System}
We proposed a type system that gives a proof theoretic interpretation for LP, where Horn formulas correspond to the notion of type, and a successful query yields a first order proof term. The type system also provided us with  a precise tool to show that realizability transformation preserves both proof-theoretic meaning of the program and the execution behaviour of the unification reduction. % In previous work \cite{JKK15}, the rewriting tree in a sense reifies all the possible term-matching reduction paths, there tree-form allow better corecursive pattern analysis, the proof content of the formula is implicit, our proof system gives us a more direct way to study both the proof theoretic aspect and operational aspect of LP. 

%\noindent \textbf{The relation of LP-Struct and LP-Unif} 

We formulated S-resolution as LP-Struct reduction, which can be seen as a reduction strategy that combines term-matching reduction with substitutional reduction.
This formulation allowed us to study the operational relation between LP-Struct and LP-Unif. %, which is novel to our knowlege.
The operational equivalence of LP-Struct and LP-Unif is by no means obvious.
Previous work (\cite{JKK15}, \cite{komendantskaya2014})
only gives soundness and completeness of LP-Struct with respect to the Herbrand model.
 We identified that productivity and non-overlapping are essential for showing their operational equivalence.
%condition, and study the relation of LP-Struct and LP-Unif under these conditions.
%Establishing their operational equivalnece allowed to
Therefore, these two properties identify
the ``structural'' fragment of logic programs.
% -- that is, those bearing these two properties.
%% Evident by Example \ref{ex:conn}, without realizablity transformation, for query $\mathrm{Connect(\mathrm{node}_1, \mathrm{node}_3)}$,
%% LP-Unif can find a success path, while LP-Struct will necessarily diverge. Only with the help of realizability transformation, we are able to show LP-Struct and LP-Unif are operational equivalence. 

%\noindent \textbf{About Realizability Transformation} 

Realizability transformation proposed here ensures that the resulting programs are productive and non-overlapping. % -- the properties we show to be crucial for
%S-resolution via LP-Struct, and for equivalent derivations by SLD-resolution via LP-Unif.
It  preserves the proof-theoretic meaning of the program, in a formally defined sense of Theorems~\ref{th6}-\ref{th8}.
It serves as a proof-method that enables us to show the operational equivalence of LP-Unif and LP-Struct, for productive and non-overlapping programs.
It is general, applies to any logic program, and can be easily mechanised. Finally, it allows to automatically record the proof content in the course of reductions, as Theorem~\ref{th7} establishes.

With the proof system for LP-reductions we proposed, we are planning to further investigate the interaction of LP-TM/Unif/Struct with typed functional languages. We expect to find a tight connection
between our work and the type class inference, cf.~(\cite{wadler1989make,jones2003qualified}).
Using terminology of this paper, a type class corresponds to an atomic formula, an instance declaration corresponds to a Horn formula, and the instance resolution process in type class inference uses
LP-TM reductions, in which evidence for the type class corresponds to our notion of proof.
Realizability transformation then gives a method to record the proof automatically.
A careful examination of these connections is warranted. %% We also want to study how exactly LP-Struct and LP-Unif can be used to accommondate the notion of constraints in the typed functional languages. Constraints in a typed functional languages are often in the form of \textit{generalized algebraic data types}(\cite{cheney2003first}, \cite{augustsson1994silly}), there the shape of constaints is the equality on types, and the type checking process will need to generate evidence for equality constaints, again it seems to fit well with the proof system
%% we proposed. 

%\noindent \textbf{Future Work}
If one works only with Horn-formulas in LP, then 
we know that the proof of a successful query can be normalized to a first order proof term. It seems that nothing interesting can happen to the proof term. But
when we plug the proof system into a typed functional language in the form of a type class and instance declaration, the proof will correspond to the evidence
for the type class, and it will interact with the underlining functional program, and eventually will be run
 as a program. For example, the following declaration specifies a way to construct equality class instance for datatype list and int:
 {\footnotesize
\begin{center}
\begin{tabular}{lll}
 $\kappa_1 : $ & & $ \mathrm{Eq}(x) \Rightarrow \mathrm{Eq}(\mathrm{list} (x))$
\\
 $\kappa_2 : $  & &  $ \Rightarrow \mathrm{Eq}(\mathrm{int})$
\end{tabular}
\end{center}
}
\noindent Here $\mathrm{list}$ is a function symbol, $\mathrm{int}$ is a constant and $x$ is variable; $\kappa_1, \kappa_2$ will be defined
as functional programs that are used to construct the evidence. When the underlining functional system
makes a query $\mathrm{Eq}(\mathrm{list} (\mathrm{int}))$, we can use LP-TM to construct
 a proof for $\mathrm{Eq}(\mathrm{list} (\mathrm{int}))$, which is $\kappa_2\ \kappa_1$, and
then $\kappa_2\ \kappa_1$ will serve as runtime evidence for the corresponding method, thus yielding computational meaning of the proof.

\bibliographystyle{plain}
 
% The bibliography should be embedded for final submission.
\bibliography{tm-lp} 
\newpage
\appendix

\section{Proof of Theorem \ref{real:sn} and \ref{fst}}
We are going to prove a nontrival property about the type system that we just set up. 
The proof is a simplification of Tait-Girard's reducibility method. 

\begin{definition}[Reducibility Set]
Let $N$ denotes the set of all strong normalizing proof terms. We define reducibility set $\mathsf{RED}_{F}$ by induction on structure of $F$: 

  \begin{itemize}
  \item $p \in \mathsf{RED}_{A_1,..., A_n \Rightarrow B}$ with $n \geq 0$ iff for any $p_i \in N$, $p \ p_1 \ ... \ p_n \in N$. 
  \item $p \in \mathsf{RED}_{\forall \underline{x}. A_1,..., A_n \Rightarrow B}$ iff $p \in \mathsf{RED}_{A_1,..., A_n \Rightarrow B}$.
  \end{itemize}
\end{definition}

\begin{lemma}
  \label{irr}
  $\mathsf{RED}_{\underline{A} \Rightarrow B} = \mathsf{RED}_{\phi \underline{A} \Rightarrow \phi B}$.
\end{lemma}
\begin{lemma}
  If $p \in \mathsf{RED}_F$, then $p \in N$. 
\end{lemma}

\begin{proof}
By Induction on $F$:
\begin{itemize}      
  \item Base Case: $F$ is of the form $A_1 , ... , A_n \Rightarrow B$. By definition, $p\ p_1\ ... \ p_n \in N$ for any $p_i \in N$. Thus $p \in N$. 
    \item Step Case: $F$ is of the form $\forall \underline{x} . A_1 , ... , A_n \Rightarrow B$. $p \in \mathsf{RED}_{\forall \underline{x} . A_1 , ... , A_n \Rightarrow B}$ implies $p \in \mathsf{RED}_{A_1 , ... , A_n \Rightarrow B}$. Thus by IH, $p \in N$. 
\end{itemize}
\end{proof}

\begin{lemma}
  \label{strong}
  If $e : F$, $e \in \mathsf{RED}_{F}$. 
\end{lemma}
\begin{proof}
  By induction on derivation of $e : F$.
  \begin{itemize}
  \item Base Case:

    \begin{tabular}{l}
\infer{\kappa : \forall \underline{x} . \Rightarrow B}{}
\end{tabular}

This case $\kappa \in N$. 
  \item Base Case:

    \begin{tabular}{l}
\infer{\kappa : \forall \underline{x} . A_1,..., A_n \Rightarrow B}{}
\end{tabular}

Since $\kappa$ is a constant, thus for any $p_i \in N$, $\kappa\ p_1\ ...\ p_n \in N$. So
$\kappa \in \mathsf{RED}_{ A_1,..., A_n \Rightarrow B}$, thus $\kappa \in \mathsf{RED}_{\forall \underline{x}. A_1,..., A_n \Rightarrow B}$.
   \item Step Case: 
\

     \begin{tabular}{l}
\infer[cut]{\lambda \underline{a} . \lambda \underline{b} . (e_2\ \underline{b})\ (e_1\ \underline{a}) : \underline{A}, \underline{B} \Rightarrow C}{e_1 : \underline{A} \Rightarrow D & e_2 : \underline{B}, D \Rightarrow C}
\end{tabular}

We need to show $\lambda \underline{a} . \lambda \underline{b} . (e_2\ \underline{b})\ (e_1\ \underline{a}) \in \mathsf{RED}_{  \underline{A},  \underline{B} \Rightarrow   C}$. By IH, we know that $(e_2\ \underline{b})\ (e_1\ \underline{a}) \in N$. Let $p_1 \in N,...,p_n \in N, q_1 \in N,..., q_m \in N$. We
are going to show for any $e$ with $(\lambda \underline{a} . \lambda \underline{b} . (e_2\ \underline{b})\ (e_1\ \underline{a})) \ \underline{p}\ \underline{q} \to_\beta e$, then $e \in N$. We proceed by induction on $(\nu((e_2\ \underline{b})\ (e_1\ \underline{a})), \nu(\underline{p}), \nu(\underline{q}))$, where $\nu$ is a function to get the length of the reduction path to normal form.
\begin{itemize}
\item Base Case: $(\nu((e_2\ \underline{b})\ (e_1\ \underline{a})), \nu(\underline{p}), \nu(\underline{q})) = (0,..., 0)$. 
The only reduction possible is $(\lambda \underline{a} . \lambda \underline{b} . (e_2\ \underline{b})\ (e_1\ \underline{a})) \ \underline{p}\ \underline{q} \to_\beta (e_2\ \underline{q})\ (e_1\ \underline{p})$. We know that $(e_2\ \underline{q})\ (e_1\ \underline{p}) \in N$.
\item Step Case: There are several possible reductions, but all will decrease $(\nu((e_2\ \underline{b})\ (e_1\ \underline{a})), \nu(\underline{p}), \nu(\underline{q}))$, thus we conclude that by induction hypothesis.
\end{itemize}
So $\lambda \underline{a} . \lambda \underline{b} . (e_2\ \underline{b})\ (e_1\ \underline{a}) \in N$. 
\item Step Case: 

  \begin{tabular}{l}
\infer[inst]{e : [\underline{t}/\underline{x}]F}{e : \forall \underline{x} . F}
\end{tabular}

By IH, we konw that $e \in \mathsf{RED}_{\forall \underline{x} . F}$, so by definition we know that $e \in \mathsf{RED}_{F}$. By Lemma \ref{irr}, $e \in \mathsf{RED}_{[\underline{t}/\underline{x}]F}$.
\item Step Case: 

  \begin{tabular}{l}
\infer[gen]{e: \forall \underline{x} . F}{e : F}
\end{tabular}

By IH, we know that $e \in \mathsf{RED}_{F}$, so we know that $e \in \mathsf{RED}_{\forall \underline{x} . F}$. 
  \end{itemize}
\end{proof}

\begin{theorem}[Strong Normalization]
  If $e : F$, then $e \in N$.
\end{theorem}
\begin{proof}
  By Lemma \ref{strong}.
\end{proof}
\begin{definition}[First Orderness]
  We say $p$ is first order inductively:
  \begin{itemize}
  \item A proof term variable $a$ or proof term constant $\kappa$ is first order.
  \item if $n, n'$ are first order, then $n\ n'$ is first order.
  \end{itemize}
\end{definition}

\begin{lemma}
  \label{fo:sub}
  If $n, n'$ are first order, then $[n'/a]n$ is first order. 
\end{lemma}

\begin{lemma}
\label{FO}
  If $e : [\forall \underline{x}.] \underline{A} \Rightarrow B$, then either $e$ is a proof term constant or it is normalizable to the form $\lambda \underline{a}. n$, where $n$ is first order normal term. 
\end{lemma}
\begin{proof}
  By induction on the derivation of $e : [\forall \underline{x}.] \underline{A} \Rightarrow B$.
  \begin{itemize}
  \item Base Cases: Axioms, in this case $e$ is a proof term constant.
    \item Step Case: 
      
           \begin{tabular}{l}
\infer[cut]{\lambda \underline{a} . \lambda \underline{b} . (e_2\ \underline{b})\ (e_1\ \underline{a}) : \underline{A}, \underline{B} \Rightarrow C}{e_1 : \underline{A} \Rightarrow D & e_2 : \underline{B}, D \Rightarrow C}
\end{tabular}

           By IH, we know that $e_1 = \kappa$ or $e_1 = \lambda \underline{a}.n_1$; $e_2 = \kappa'$ or $e_2 = \lambda \underline{b} d . n_2$. We know that $e_1 \underline{a}$ will be normalizable to a first order proof term. And $e_2 \underline{b}$ will be normalized to either $\kappa' \underline{b}$ or $\lambda d . n_2$. So by Lemma \ref{fo:sub}, we conclude that 
           $\lambda \underline{a} . \lambda \underline{b} . (e_2\ \underline{b})\ (e_1\ \underline{a})$ is normalizable to $\lambda \underline{a} . \lambda \underline{b} . n$ for some
           first order normal term $n$.
           \item The other cases are straightforward.
  \end{itemize}
\end{proof}
\begin{theorem}
  If $e : [\forall \underline{x}.] \Rightarrow B$, then $e$ is normalizable to a first order proof term.
\end{theorem}
\begin{proof}
  By lemma \ref{FO}, subject reduction and strong normalization theorem. 
\end{proof}
\section{Proof of Theorem \ref{realI}}

\begin{theorem}
Given axioms $\Phi$, if $e: [\forall \underline{x}] . \underline{A}\Rightarrow B$ holds with $e$ in normal form, then $F(e : [\forall \underline{x}] . \underline{A}\Rightarrow B)$ holds for axioms $F(\Phi)$.
\end{theorem}
\begin{proof}
  By induction on the derivation of $e: [\forall \underline{x}] . \underline{A}\Rightarrow B$.
  \begin{itemize}
  \item Base Case: 

   \

    \begin{tabular}{l}
     \infer{\kappa : \forall \underline{x} . \underline{A}\Rightarrow B}{}      

    \end{tabular}
   
\

 In this case, we know that $F(\kappa : \forall \underline{x} . \underline{A}\Rightarrow B) = \kappa : \forall \underline{x}. \forall \underline{y} . A_1[y_1], ..., A_n[y_n] \Rightarrow B[f_\kappa (y_1,..., y_n)] \in F(\Phi)$. 

\item Step Case: 

\

\begin{tabular}{l}
\infer[cut]{\lambda \underline{a} . \lambda \underline{b} . (e_2\ \underline{b})\ (e_1\ \underline{a}) : \underline{A}, \underline{B} \Rightarrow C}{e_1 : \underline{A} \Rightarrow D & e_2 : \underline{B}, D \Rightarrow C}
  
\end{tabular}

\

We know that the normal form of $e_1$ must be $\kappa_1$ or $\lambda \underline{a}. n_1$; the normal form of $e_1$ must be $\kappa_2$ or $\lambda \underline{b} d. n_2$, with $n_1, n_2$ are first order. 
\begin{itemize}
\item $e_1 \equiv \kappa_1, e_2 \equiv \kappa_2$. By IH, we know that $F(\kappa_1 : \underline{A} \Rightarrow D) = \kappa_1 : A_1[y_1],..., A_1[y_1] \Rightarrow D[f_{\kappa_1}(y_1,..., y_n)] $ and 
$F(\kappa_2 : \underline{B}, D \Rightarrow C) = \kappa_2 : B_1[z_1],..., B_m[z_m], D[y] \Rightarrow C[f_{\kappa_2}(z_1,..., z_m, y)]$ hold. So by \textit{gen} and \textit{inst}, we have

\noindent $\kappa_2 : B_1[z_1],..., B_m[z_m], D[f_{\kappa_1}(y_1,..., y_n)] \Rightarrow C[f_{\kappa_2}(\underline{z}, f_{\kappa_1}(\underline{y}))]$. 

\noindent Then by the cut rule, we have 

\noindent $\lambda \underline{a} . \lambda \underline{b} . \kappa_2 \underline{b} (\kappa_1 \underline{a}) : A_1[y_1],..., A_1[y_1],  B_1[z_1],..., B_m[z_m] \Rightarrow C[f_{\kappa_2}(\underline{z}, f_{\kappa_1}(\underline{y}))]$. We 
can see that $\interp{\kappa_2 \underline{b} (\kappa_1 \underline{a})}_{[\underline{y}/\underline{a}, \underline{z}/\underline{b}]} = f_{\kappa_2}(\underline{z}, f_{\kappa_1}(\underline{y}))$.

\item $e_1 \equiv \lambda \underline{a}. n_1, e_2 \equiv \lambda \underline{b}d. n_2$. By IH, we know that $F(\lambda \underline{a}. n_1 : \underline{A} \Rightarrow D) = \lambda \underline{a}. n_1 : A_1[y_1],..., A_1[y_1] \Rightarrow D[\interp{n_1}_{[\underline{y}/\underline{a}]}]$ and 
$F(\lambda \underline{b} d. n_2 : \underline{B}, D \Rightarrow C) = \lambda \underline{b} d. n_2 : B_1[z_1],..., B_m[z_m], D[y] \Rightarrow C[\interp{n_2}_{[\underline{z}/ \underline{b}, y/d]}]$ hold. So by \textit{gen} and \textit{inst}, we have

\noindent $\lambda \underline{b} d. n_2 : B_1[z_1],..., B_m[z_m], D[\interp{n_1}_{[\underline{y}/\underline{a}]}] \Rightarrow C[\interp{n_2}_{[\underline{z}/\underline{b}, \interp{n_1}_{[\underline{y}/\underline{a}]}/d]}]$. 

\noindent Then by the cut rule and beta reductions, we have $\lambda \underline{a} . \lambda \underline{b} . ([n_1/d]n_2) : A_1[y_1],..., A_1[y_1],  B_1[z_1],..., B_m[z_m] \Rightarrow C[\interp{n_2}_{[\underline{z}/\underline{b}, \interp{n_1}_{[\underline{y}/\underline{a}]}/d]}]$. We 
know that $\interp{[n_1/d]n_2}_{[\underline{y}/\underline{a}, \underline{z}/\underline{b}]} = \interp{n_2}_{[\underline{z}/\underline{b}, \interp{n_1}_{[\underline{y}/\underline{a}]}/d]}$. 
\item The other cases are handle similarly. 
\end{itemize}

\item Step Case:

\

  \begin{tabular}{l}
\infer[inst]{\lambda \underline{a} . n : [\underline{t}/\underline{x}]\underline{A} \Rightarrow [\underline{t}/\underline{x}]B }{\lambda \underline{a} . n : \forall \underline{x} . \underline{A} \Rightarrow B}
\end{tabular}

\

By IH, we know that $F(\lambda \underline{a} . n : \forall \underline{x} . \underline{A} \Rightarrow B) = \lambda \underline{a} . n : \forall \underline{x}. \forall \underline{y} . A_1[y_1],..., A_n[y_n] \Rightarrow B[\interp{n}_{[\underline{y}/\underline{a}]}]$ holds for $F(\Phi)$. By \textit{Inst} rule, we instantiate $y_i$ with $y_i$,  we have 
$\lambda \underline{a} . n : [\underline{t}/\underline{x}] A_1[y_1],..., [\underline{t}/\underline{x}] A_n[y_n] \Rightarrow [\underline{t}/\underline{x}] B[\interp{n}_{[\underline{y}/\underline{a}]}]$

\item Step Case:

\

  \begin{tabular}{l}
\infer[gen]{e: \forall \underline{x} . F}{e : F}
\end{tabular}

\

This case is straightforwardly by IH. 
  \end{itemize}
\end{proof}

\section{Proof of Theorem \ref{record}}

\begin{lemma}
\label{realII}
  If $F(\Phi) \vdash \{A_1[y_1],..., A_n[y_n]\} \leadsto^*_{\gamma} \emptyset$, and $y_1,..., y_n$ are fresh, then there exists proofs $e_1 : \forall \underline{x} . \Rightarrow \gamma A_1[\gamma y_1],..., e_n : \forall \underline{x} . \Rightarrow \gamma A_n[\gamma y_n]$ with $\interp{e_i}_{\emptyset} = \gamma y_i $ given axioms $F(\Phi)$.  
\end{lemma}
\begin{proof}
  By induction on the length of the reduction. 
  \begin{itemize}
  \item Base Case. Suppose the length is one, namely, $F(\Phi) \vdash \{A[y]\} \leadsto_{\kappa, \gamma_1} \emptyset$. Thus there exists $(\kappa : \forall \underline{x} .  \Rightarrow C[f_\kappa]) \in F(\Phi)$(here $f_\kappa$ is a constant), such that $C[f_\kappa] \sim_{\gamma_1} A[y]$.  Thus $ \gamma_1 (C[f_\kappa]) \equiv \gamma_1 A[\gamma_1 y]$. So $\gamma_1 y \equiv f_\kappa$
and $\gamma_1 C \equiv \gamma_1 A$. We have $\kappa :\ \Rightarrow  \gamma_1 C[f_\kappa]$ by the \textit{inst} rule, thus $\kappa :\ \Rightarrow \gamma_1 A[\gamma_1 y]$, hence $\kappa : \forall \underline{x} . \Rightarrow \gamma_1 A[\gamma_1 y]$ by the \textit{gen} rule and $\interp{\kappa}_{\emptyset} = f_{\kappa}$.

  \item Step Case. Suppose $F(\Phi) \vdash \{A_1[y_1], ..., A_i[y_i],..., A_n[y_n]\} \leadsto_{\kappa, \gamma_1} $ 

$\{\gamma_1 A_1[y_1],...,  \gamma_1 B_1[z_1],...,  \gamma_1 B_m[z_m],..., \gamma_1 A_n[y_n]\} \leadsto^*_{\gamma} \emptyset$,

\noindent  where $\kappa : \forall \underline{x} . \forall \underline{z} . B_1[z_m],..., B_n[z_m] \Rightarrow C[f_\kappa(z_1,..., z_m)] \in F(\Phi)$,

\noindent and $C[f_\kappa(z_1,..., z_m)] \sim_{ \gamma_1} A_i[y_i]$. So we know $ \gamma_1 C[f_\kappa(z_1,..., z_m)] \equiv \gamma_1 A_i[\gamma_1 y_i]$,  $\gamma_1 y_i \equiv f_\kappa(z_1,..., z_m),  \gamma_1 C \equiv \gamma_1 A_i$ and 

\noindent $\mathrm{dom}(\gamma_1) \cap \{z_1,..., z_m, y_1,..,y_{i-1}, y_{i+1}, y_n\} = \emptyset$. By IH, we know that there exists proofs $e_1 : \forall \underline{x}. \Rightarrow \gamma \gamma_1 A_1[\gamma y_1],..., p_1 : \forall \underline{x}. \Rightarrow \gamma   \gamma_1 B_1[\gamma z_1],..., p_m : \forall \underline{x}. \Rightarrow \gamma   \gamma_1 B_m[\gamma z_m],..., e_n : \forall \underline{x} . \Rightarrow \gamma \gamma_1 A_n[\gamma y_n]$ and $\interp{e_1}_{\emptyset} = \gamma y_1, ..., \interp{p_1}_\emptyset = \gamma z_1, ..., \interp{e_n}_{\emptyset} = \gamma y_n $ . We can construct a proof $e_i = \kappa \ p_1\ ... p_m$ with $e_i : \forall \underline{x} . \Rightarrow \gamma \gamma_1 A_i[\gamma \gamma_1 y_i]$, by first use the \textit{inst} to instantiate the quantifiers of $\kappa$, then applying the cut rule $m$ times. Moreover, we have $\interp{\kappa \ p_1\ ... p_m}_{\emptyset} = f_\kappa(\interp{p_1}_\emptyset,...,\interp{p_m}_\emptyset) = \gamma (f_\kappa (z_1,..., z_m)) = \gamma \gamma_1 y_i$. 
     \end{itemize}

\end{proof}

\begin{theorem}
  Given axioms $\Phi$, suppose $F(\Phi) \vdash \{A[y]\} \leadsto^*_{\gamma} \emptyset$. We have $p : \forall \underline{x} . \Rightarrow \gamma A[\gamma y]$ where $p$ is in normal form and $\interp{p}_{\emptyset} = \gamma y$. 
\end{theorem}
\begin{proof}
By Lemma \ref{realII}.
\end{proof}

\section{Proof of Lemma \ref{sc:unif}}
\begin{lemma}
If $\Phi \vdash \{A_1,..., A_n\} \leadsto^* \emptyset$, then $F(\Phi) \vdash \{A_1[y_1],..., A_n[y_n]\} \leadsto^* \emptyset$ with $y_i$ fresh. 
\end{lemma}
\begin{proof}
  By induction on the length of reduction.
  \begin{itemize}
     \item Base Case. Suppose the length is one, namely, $\Phi \vdash \{A\} \leadsto_{\kappa, \gamma_1} \emptyset$. Then there exists $(\kappa : \forall \underline{x} .\  \Rightarrow C) \in \Phi$ such that $C \sim_{\gamma_1} A$.  Thus $\kappa : \forall \underline{x}. \ \Rightarrow C[f_\kappa] \in F(\Phi)$ and $(C[f_\kappa]) \sim_{\gamma_1[f_\kappa/y]} A[y]$. So $F(\Phi) \vdash A[y]\leadsto \emptyset$.
     \item Step Case. Suppose 
       
       \noindent $\Phi \vdash \{A_1, ..., A_i,..., A_n\} \leadsto_{\kappa, \gamma_1} \{\gamma_1 A_1,...,  \gamma_1 B_1,...,  \gamma_1 B_m,..., \gamma_1 A_n\} \leadsto^*_{\gamma} \emptyset$,

\noindent  where $\kappa : \forall \underline{x}. B_1,..., B_m \Rightarrow C \in \Phi$, $C \sim_{ \gamma_1} A_i$. So we know that 

\noindent $\kappa : \forall \underline{x}. B_1[z_1],..., B_m[z_m] \Rightarrow C[f_\kappa(\underline{z})] \in F(\Phi)$ and $C[f_\kappa(\underline{z})] \sim_{\gamma_1[f_\kappa(\underline{z})/y_i]} A_i[y_i]$. Thus $F(\Phi) \vdash \{A_1[y_1], ..., A_i[y_i],..., A_n[y_n]\} \leadsto_{\kappa, \gamma_1[f_\kappa(\underline{z})/y_i]} $ 

{\scriptsize
 $\{\gamma_1[f_\kappa(\underline{z})/y_i] A_1[y_1],...,  \gamma_1[f_\kappa(\underline{z})/y_i] B_1[z_1],...,  \gamma_1[f_\kappa(\underline{z})/y_i] B_m[z_m],..., \gamma_1[f_\kappa(\underline{z})/y_i] A_n[y_n]\}$
}
\noindent $ \equiv \{\gamma_1 A_1[y_1],...,  \gamma_1 B_1[z_1],...,  \gamma_1 B_m[z_m],..., \gamma_1 A_n[y_n]\}$. By IH, 

\noindent $F(\Phi) \vdash \{\gamma_1 A_1[y_1],...,  \gamma_1 B_1[z_1],...,  \gamma_1 B_m[z_m],..., \gamma_1 A_n[y_n]\} \leadsto^* \emptyset$.

  \end{itemize}
\end{proof}
\begin{lemma}
If $F(\Phi) \vdash \{A_1[y_1],..., A_n[y_n]\} \leadsto^* \emptyset$, then $\Phi \vdash \{A_1,..., A_n\} \leadsto^* \emptyset$. 
\end{lemma}
\begin{proof}
  By induction on the length of reduction.
  \begin{itemize}
     \item Base Case. Suppose the length is one, namely, $F(\Phi) \vdash \{A[y]\} \leadsto_{\kappa, \gamma_1} \emptyset$. Thus there exists $(\kappa : \forall \underline{x} .  \Rightarrow C[f_\kappa]) \in F(\Phi)$ such that $C[f_\kappa] \sim_{\gamma_1} A[y]$.  Thus $  C \sim_{\gamma_1-[f_\kappa/y]} A$. So $\Phi \vdash A\leadsto \emptyset$.
     \item Step Case. Suppose $F(\Phi) \vdash \{A_1[y_1], ..., A_i[y_i],..., A_n[y_n]\} \leadsto_{\kappa, \gamma_1} $ 

$\{\gamma_1 A_1[y_1],...,  \gamma_1 B_1[z_1],...,  \gamma_1 B_m[z_m],..., \gamma_1 A_n[y_n]\} \leadsto^*_{\gamma} \emptyset$,

\noindent  where $\kappa : \forall \underline{x} . \forall \underline{z} . B_1[z_m],..., B_m[z_m] \Rightarrow C[f_\kappa(z_1,..., z_m)] \in F(\Phi)$,

\noindent and $C[f_\kappa(z_1,..., z_m)] \sim_{ \gamma_1} A_i[y_i]$. So we know $C \sim_{\gamma_1-[f_\kappa(\underline{z})/y_i]} A_i$. Let $\gamma = \gamma_1-[f_\kappa(\underline{z})/y_i]$. We have

\noindent $\Phi \vdash \{A_1,..., A_i,..., A_n\} \leadsto \{\gamma A_1,..., \gamma B_1,..., \gamma B_m, ..., \gamma A_n\}$

$\equiv \{\gamma_1 A_1,..., \gamma_1 B_1,..., \gamma_1 B_m, ..., \gamma_1 A_n\}$. By IH,
we know

\noindent $\Phi \vdash \{\gamma_1 A_1,..., \gamma_1 B_1,..., \gamma_1 B_m, ..., \gamma_1 A_n\} \leadsto^* \emptyset$.

  \end{itemize}
  
\end{proof}

\section{Proof of Theorem \ref{ortho:equiv}}

\begin{lemma}
  \label{tm-to-unif}
If $\Phi \vdash \{D_1,..., D_i,..., D_n\} \to_{\kappa, \gamma} \{D_1,.., \sigma E_1,..., \sigma E_m,..., D_n\}$, with $\kappa : \forall \underline{x}. \underline{E} \Rightarrow C \in \Phi$ and $C \mapsto_\sigma D_i$ for any $\gamma$, then $\Phi \vdash \{D_1,..., D_i, ..., D_n\} \leadsto_{\kappa, \gamma} $

\noindent $\{D_1,.., \sigma E_1,..., \sigma E_m,..., D_n\}$. 
\end{lemma}
\begin{proof}
  Since for $\Phi \vdash \{D_1,..., D_i,..., D_n\} \to_{\kappa, \gamma} \{D_1,.., \sigma E_1,..., \sigma E_m,..., D_n\}$, with $\kappa : \forall \underline{x}. \underline{E} \Rightarrow C \in \Phi$ and $C \mapsto_\sigma D_i$, we have $\Phi \vdash \{D_1,..., D_i, ..., D_n\} \leadsto_{\kappa, \sigma \cdot \gamma} \{\sigma D_1,.., \sigma E_1,..., \sigma E_m,..., \sigma D_n\}$. But $\mathrm{dom}(\sigma) \in \mathrm{FV}(C)$, thus we have
  
  \noindent $\Phi \vdash \{D_1,..., D_i, ..., D_n\} \leadsto_{\kappa, \gamma} \{D_1,.., \sigma E_1,..., \sigma E_m,..., D_n\}$. 

\end{proof}

\begin{lemma}\label{struct-to-unif}
  \
  
\noindent Given $\Phi$ is non-overlapping, if $\Phi \vdash \{A_1,..., A_n\}  (\hookrightarrow_{\kappa, \gamma} \cdot \to^\nu_\gamma) \{C_1, ..., C_m\}$, then $\Phi \vdash \{A_1, ..., A_n\} \leadsto^*_\gamma \{C_1, ..., C_m\}$. 
\end{lemma}
\begin{proof}
  Given $\Phi \vdash \{A_1,..., A_n\} (\hookrightarrow_{\kappa, \gamma} \cdot \to^\nu_\gamma) \{C_1, ..., C_m\}$, we know the actual reduction path must be of the form $\Phi \vdash \{A_1,..., A_n\}\hookrightarrow_{\kappa, \gamma} \{\gamma A_1,..., \gamma A_n\} \to_{\kappa, \gamma} \{\gamma A_1, ..., \gamma B_1, ..., \gamma B_n,..., \gamma A_n\} \to^\nu_{\gamma} \{C_1, ..., C_m\}$. Note that $\gamma$ is unchanged along the term-matching reduction. The $\to$ following right after $\hookrightarrow$ can not use a different rule other than $\kappa$, it would mean $\gamma A_i \equiv \gamma C$ with $\kappa : \forall \underline{x}.\underline{B} \Rightarrow C \in \Phi$ and $A_i \equiv \sigma B$ with $\kappa' : \forall \underline{x}. \underline{D} \Rightarrow B \in \Phi$. This implies $\gamma C \equiv \gamma \sigma B$, contradicting
the non-overlapping restriction. Thus we have $\Phi \vdash \{A_1,..., A_n\}\leadsto_{\kappa,  \gamma} \{\gamma A_1, ..., \gamma B_1, ..., \gamma B_n,..., \gamma A_n\}$. By Lemma \ref{tm-to-unif},
we have $\Phi \vdash \{A_1,..., A_n\}\leadsto_{\kappa,  \gamma} \{\gamma A_1, ..., \gamma B_1, ..., \gamma B_n,..., \gamma A_n\} \leadsto^*_\gamma \{C_1, ..., C_m\}$

\end{proof}

\begin{lemma}
  Given $\Phi$ is non-overlapping, if $\Phi \vdash \{A_1,..., A_n\} (\to^\mu \cdot \hookrightarrow^1)^*_\gamma  \{C_1, ..., C_m\}$ with $\{C_1, ..., C_m\}$ in $\to^\mu \cdot \hookrightarrow^1$-normal form, then $\Phi \vdash \{A_1, ..., A_n\} \leadsto^*_\gamma \{C_1, ..., C_m\}$ with $\{C_1, ..., C_m\}$ in $\leadsto$-normal form. 
\end{lemma}
\begin{proof}
  Since $\Phi \vdash \{A_1, ..., A_n\}  (\to^\mu \cdot \hookrightarrow^1)^*_\gamma \{C_1, ..., C_m\}$, this means the reduction path must be of the form $\Phi \vdash \{A_1,..., A_n\} \to^\nu \cdot \hookrightarrow^1 \cdot \to^\nu \cdot \hookrightarrow^1 ... \to^\nu \cdot \hookrightarrow^1 \cdot \to^\nu  \{C_1, ..., C_m\}$. Thus $\Phi \vdash \{A_1,..., A_n\} \to^\nu \cdot (\hookrightarrow^1 \cdot \to^\nu) \cdot (\hookrightarrow^1 ... \to^\nu) \cdot (\hookrightarrow^1 \cdot  \to^\nu)  \{C_1, ..., C_m\}$. By Lemma \ref{tm-to-unif} and Lemma \ref{struct-to-unif}, we have $\Phi \vdash \{A_1, ..., A_n\} \leadsto^*_\gamma \{C_1, ..., C_m\}$ with $\{C_1, ..., C_m\}$ in $\leadsto$-normal form.
  
  %% We want to show that each $\hookrightarrow \cdot \to$ in the sequence $ \to^\nu \cdot \hookrightarrow^1 ... \to^\nu \cdot \hookrightarrow^1  \to^\nu$ can be emulated by $\leadsto$, i.e. if $\Phi \vdash \{A_1,...A_m\} \hookrightarrow_{\kappa, \gamma} {}$the $\to$ must use the
%% same rule as $\hookrightarrow$, this is the case due to the non-overlapping restriction. Suppose the $\to$ following $\hookrightarrow$ uses different rules, it means $\gamma A \equiv \gamma C$ with $\underline{E} \Rightarrow C \in \Phi$ and $A \equiv \sigma B$ with $\underline{D} \Rightarrow B \in \Phi$. This implies $\gamma C \equiv \gamma \sigma B$, contradicting
%% the non-overlapping restriction. So the $\to$ must use the same rule as $\hookrightarrow$,
%% thus $\hookrightarrow \cdot \to$ can be emulated by $\leadsto$. And for each $\to$ that is not on the right
%%  of $\hookrightarrow$, can use $\leadsto$ to emulate $\to$. So we have $\Phi \vdash \{A\} \leadsto^*_\gamma \emptyset$
\end{proof}

\begin{lemma}
  Given $\Phi$ is a non-overlapping, if $\Phi \vdash \{A_1, ..., A_n\} \leadsto^*_\gamma \{C_1, ..., C_m\}$ with $\{C_1, ..., C_m\}$ in $\leadsto$-normal form , then $\Phi \vdash \{A_1, ..., A_n\} (\to^\nu \cdot \hookrightarrow^1)^*_\gamma \{C_1, ..., C_m\}$ with $\{C_1, ..., C_m\}$ in $\to^\nu \cdot \hookrightarrow^1$-normal form. 
\end{lemma}
\begin{proof}
  By induction on the length of $\leadsto^*_\gamma$. 

\begin{itemize}
\item Base Case: $\Phi \vdash \{A_1,...,A_i, ..., A_n\} \leadsto_{\kappa, \gamma} \{\gamma A_1,...,\gamma B_1, ..., \gamma B_m ..., \gamma A_n \}$ with $\kappa : \forall \underline{x}. \ \underline{B}\Rightarrow C \in \Phi$, $C \sim_\gamma A_i$ and $\{\gamma A_1,...,\gamma B_1, ..., \gamma B_m ..., \gamma A_n \}$ in $\leadsto$-normal form . We have $\Phi \vdash \{A_1,...,A_i, ..., A_n\} \hookrightarrow_{\kappa, \gamma} \{\gamma A_1,...,\gamma A_i, ..., \gamma A_n \} \to_\kappa \{\gamma A_1,...,\gamma B_1, ..., \gamma B_m ..., \gamma A_n\}$ with $\{\gamma A_1,...,\gamma B_1, ..., \gamma B_m ..., \gamma A_n\}$ in $\to^\nu\cdot \hookrightarrow$-normal form. Note that there can not be another
$\kappa' : \forall \underline{x}. \underline{B} \Rightarrow C' \in \Phi$ such that $\sigma C' \equiv A_i$, since this would means $\gamma C \equiv \gamma A_i \equiv \gamma \sigma C'$, violating the non-overlapping requirement.

\item Step Case: $\Phi \vdash \{A_1, ..., A_i, ..., A_n\} \leadsto_{\kappa, \gamma} \{\gamma A_1,..., \gamma B_1, ..., \gamma B_l, ..., \gamma A_n\} \leadsto^*_{\gamma'} \{C_1, ..., C_m\}$ with $\kappa : \forall \underline{x}. B_1,..., B_l \Rightarrow C \in \Phi$ and $C \sim_\gamma A_i$. 

\noindent We have $\Phi \vdash \{A_1, ..., A_i, ..., A_n\} \hookrightarrow_{\kappa, \gamma} \{\gamma A_1, ..., \gamma A_i, ..., \gamma A_n\}\to $

\noindent $\{\gamma A_1,..., \gamma B_1, ..., \gamma B_m, ..., \gamma A_n\}$. By the non-overlapping requirement, there can not be another $\kappa' : \forall \underline{x}. \underline{D} \Rightarrow C' \in \Phi$ such that $\sigma C' \equiv A_i$. 

\noindent By IH, we know $\Phi \vdash \{\gamma A_1,..., \gamma B_1, ..., \gamma B_m, ..., \gamma A_n\} (\to^\nu \cdot \hookrightarrow)_{\gamma'}^* \{C_1, ..., C_m\}$. Thus we conclude that $\Phi \vdash \{A_1, ..., A_i, ..., A_n\} (\hookrightarrow \cdot \to)^*_{\gamma'} \{C_1, ..., C_m\}$. 

\end{itemize}
% We have $\Phi \vdash \{A\} \hookrightarrow_\gamma \cdot \to \emptyset$. Note that there can not be another
% $\kappa' : \forall \underline{x}. \underline{B} \Rightarrow C' \in \Phi$ such that $\sigma C' \equiv A$, since this would means $\gamma C \equiv \gamma A \equiv \gamma \sigma C'$, violating the non-overlapping requirement.

\end{proof}

\section{Proof of Theorem \ref{prod-non-overlap}}
We assume a non-overlapping and productive program $\Phi$ in this section.

\begin{lemma}
  If $\Phi \vdash \{A_1,..., A_n\} \leadsto \{B_1,..., B_m\}$, then $\Phi \vdash \{A_1,..., A_n\} (\to^\nu \cdot \hookrightarrow^1)^* \{C_1,..., C_l\}$ and $\Phi \vdash \{B_1,..., B_m\} \to^* \{C_1,..., C_l\}$.
\end{lemma}
\begin{proof}
  Suppose $\Phi \vdash \{A_1,..., A_n\} \leadsto_{\kappa, \gamma} \{\gamma A_1,...,\gamma E_1, ..., \gamma E_l,..., \gamma A_n\}$, with $\kappa : \underline{E} \Rightarrow D \in \Phi$ and $D \sim_\gamma A_i$. Suppose $D \not \mapsto_\gamma A_i$. In this case, we have $\Phi \vdash \{A_1,..., A_n\} \hookrightarrow_{\kappa, \gamma} \cdot \to_{\kappa, \gamma} \{\gamma A_1,...,\gamma E_1, ..., \gamma E_q,..., \gamma A_n\} \to^\nu_\gamma \{C_1,..., C_l\}$. Suppose $D \mapsto_\gamma A_i$, we have $\Phi \vdash \{A_1,..., A_n\} \to_{\kappa, \gamma} \{\gamma A_1,...,\gamma E_1, ..., \gamma E_q,..., \gamma A_n\} \to^\nu_\gamma \{C_1,..., C_l\}$.
\end{proof}
\begin{lemma}\label{I}
  If $\Phi \vdash \{A_1,..., A_n\} \hookrightarrow_{\kappa,\gamma} \{\gamma A_1,..., \gamma A_n\} \to^\nu_\gamma \{B_1,..., B_m\}$, then $\Phi \vdash \{A_1,..., A_n\} \leadsto^*_\gamma \{B_1,..., B_m\}$. \end{lemma}
\begin{proof}
Suppose $\Phi \vdash \{A_1,..., A_n\} \hookrightarrow_{\kappa,\gamma} \{\gamma A_1,..., \gamma A_n\} \to^\nu_\gamma \{B_1,..., B_m\}$, we have $\Phi \vdash \{A_1,..., A_n\} \hookrightarrow_{\kappa,\gamma} \{\gamma A_1,..., \gamma A_n\} \to_\kappa \{\gamma A_1,..., \gamma C_1,..., \gamma C_l. \gamma A_n\}\to^\nu_\gamma \{B_1,..., B_m\}$ with $\kappa : \underline{C} \Rightarrow D \in \Phi$ and $D \sim_\gamma A_i$. Thus we have $\Phi \vdash \{A_1,..., A_n\} \leadsto_{\kappa,\gamma}  \{\gamma A_1,..., \gamma C_1,..., \gamma C_l,..., \gamma A_n\}$. By Lemma \ref{tm-to-unif}, we have $\Phi \vdash \{A_1,..., A_n\} \leadsto_{\kappa,\gamma}  \{\gamma A_1,..., \gamma C_1,..., \gamma C_l, ...., \gamma A_n\} \leadsto^*_\gamma \{B_1,..., B_m\}$.
\end{proof}
\begin{lemma}
If $\Phi \vdash \{A_1,..., A_n\} (\to^\nu \cdot \hookrightarrow^1)^*_\gamma \{B_1,..., B_m\}$, then $\Phi \vdash \{A_1,..., A_n\} \leadsto^*_\gamma \{B_1,..., B_m\}$. 
\end{lemma}
\begin{proof}
  By Lemma \ref{I}.
\end{proof}
\end{document}